\documentclass[journal]{IEEEtran}
\IEEEoverridecommandlockouts

\usepackage{cite}
\usepackage{amsmath,amssymb,amsfonts}
\usepackage{algorithm}
\usepackage{algpseudocode}
\usepackage{subcaption}
\usepackage{booktabs} 
\usepackage{makecell} 
\usepackage{graphicx}
\usepackage{caption}
\usepackage[compact]{titlesec}  
\titlespacing{\section}{2pt}{2pt}{2pt}
\usepackage{textcomp}
\usepackage{xcolor}
\usepackage{comment}
\usepackage{hyperref}
\def\BibTeX{{\rm B\kern-.05em{\sc i\kern-.025em b}\kern-.08em
    T\kern-.1667em\lower.7ex\hbox{E}\kern-.125emX}}
\usepackage{bm}
\usepackage{cases}

\usepackage{algorithm}
\usepackage{algpseudocode}
\usepackage{amssymb}
\usepackage{amsfonts}
\usepackage{bm}
\usepackage{amsmath}
\usepackage{amsthm}
\usepackage[noabbrev,capitalise,nameinlink]{cleveref}
\usepackage{float}
\usepackage{graphicx}
\usepackage{xcolor}
\usepackage{epstopdf}
\usepackage{multirow}
\usepackage{caption,subcaption}
 \usepackage{enumitem}
\usepackage{etoolbox}

\newtheorem{lemma}{Lemma}[section]

\newtheorem{remark}{Remark}[section]
\newtheorem{assumption}{Assumption}[section]

\begin{document}

\title{Distributed Quantum Learning over Near-term Devices: Convergence Analysis and Security Design}

\author{\IEEEauthorblockN{Atit Pokharel*, Shaba Shaon*, Thomas Morris, and Dinh C. Nguyen}

\thanks{Atit Pokharel, Shaba Shaon, Thomas Morris, and Dinh C. Nguyen are with the Department of Electrical and Computer Engineering, University of Alabama in Huntsville, USA. e-mails:\{ap1284, ss0670, tommy.morris, dinh.nguyen\}@uah.edu. *Atit and Shaba equally contribute to this work.}
\IEEEcompsocitemizethanks{*This work was supported in part by the U.S. National Science Foundation (NSF) under Grant SaTC-2513164.}
}

\maketitle

\begin{abstract}
Distributed quantum learning (DQL) has emerged as a promising paradigm to scale quantum-enhanced machine learning by interconnecting multiple quantum devices. However, for efficient real-world deployment, it is essential to characterize how DQL converges under practical scenarios while simultaneously safeguarding multi-device quantum infrastructures from evolving security threats. Addressing these aspects in an integrated manner is key to ensuring both performance and resilience in large-scale DQL systems. Therefore, this paper presents a new DQL study where our innovation lies in: (i) conducting a holistic convergence analysis for DQL under practical settings, i.e., partial device participation, non-convex loss functions, and heterogeneous data distributions, (ii) developing a novel multi-layered post-quantum cryptographic architecture with a quantum neural network-powered adaptive mechanism that monitors conditions, evaluates threats, and adjusts parameters across three National Institute of Standards and Technology (NIST)-compliant levels. \textcolor{black}{Our theoretical framework and empirical validation reveal two key insights: (i) the derived convergence bound uncovers a fundamental trade-off between convergence rate, measurement shots, and the size of the participating device subset; and (ii) findings from our evaluations on a physical testbed modeling quantum control architectures expose the performance limitations of static post-quantum security}, while confirming that our adaptive framework effectively mitigates these overheads to preserve overall system efficiency. \textcolor{black}{Specifically, the hardware experiments demonstrate that our dynamic security mechanism reduces total security execution time by approximately 49\% relative to static high-security baselines, while maintaining a threat detection accuracy of over 91\%. Furthermore, extensive simulations validate our theoretical analysis, showing strong agreement between predicted and observed convergence trends.} The coupling of these two stages ensures both theoretically grounded performance optimization and adaptive threat resilience, enabling efficient, secure, and scalable DQL deployment.


\end{abstract}


\section{Introduction}
Distributed quantum computing (DQC) has emerged as a promising paradigm to overcome the limitations of single quantum processors by interconnecting multiple nodes to perform joint computations. 
\textcolor{black}{Leveraging the principles of quantum mechanics, such as superposition, entanglement, and quantum teleportation, DQC interconnects multiple near-term quantum processors (e.g., trapped-ion or superconducting nodes) to execute workloads that exceed a single device’s qubit count or connectivity \cite{main2025distributed}. This enables the distributed implementation of variational algorithms, such as variational quantum eigensolvers (VQEs) and quantum neural networks (QNNs), making it feasible to scale domain-specific tasks in quantum chemistry \cite{jones2024distributed}, healthcare \cite{qu2025daqfl}, network optimization \cite{wu2023qucomm}, etc.}   
\textcolor{black}{This approach not only scales computational capacity but also mitigates near-term hardware constraints by distributing the computational workload across multiple quantum processing units (QPUs), so that each node executes only a sub-circuit under its local qubit count and coherence limits \cite{wu2023qucomm, ferrari2023modular}.}
 By integrating quantum networking with distributed processing, DQC paves the way for solving computational tasks beyond the reach of classical high-performance computing.

Building on the foundation of DQC, distributed quantum learning (DQL) extends these capabilities to the machine learning domain, enabling collaborative training of quantum machine learning (QML) models across spatially separated devices. Analogous to distributed machine learning in the classical domain, DQL coordinates variational quantum circuits (VQCs) over multiple nodes, each processing local datasets. This decentralized training approach addresses challenges in data locality, enhances scalability, and enables cross-organization model training without direct data sharing. Applications span a broad range of domains, including drug discovery, financial modeling, climate prediction, and autonomous systems, where both quantum speedups and privacy preservation are crucial. 

\subsection{Motivation: Early Prototypes of DQC and DQL}
A recent study \cite{main2025distributed} demonstrated the first photonic-network-based DQC system, interconnecting two trapped-ion quantum processors via a two-metre optical link to execute non-local quantum gates. Using deterministic gate teleportation,  the system achieved high-fidelity remote entangling operations and successfully executed a distributed version of Grover’s search algorithm involving multiple non-local two-qubit gates, marking a significant milestone in performing large-scale quantum circuits across spatially separated processors without compromising connectivity. Complementing these advances, another recent work \cite{bhatia2024communication} showcased a communication-efficient DQL framework for multi-center healthcare, enabling multiple institutions to collaboratively train variational quantum models without sharing raw patient data. \textcolor{black}{The authors used quantum natural gradient descent (QNGD) for the client-side optimization of variational circuits in their distributed setup with limited qubits and compared it against standard gradient descent under unbalanced client data.} 
\textcolor{black}{Collectively, the successful hardware demonstration of non-local entanglement in \cite{main2025distributed} supports the {physical feasibility} of modular architectures in distributed quantum computation. In addition, the multi-site distributed training results in \cite{bhatia2024communication} demonstrate practical benefits in distributed quantum learning with data locality for privacy-sensitive settings and robustness under unbalanced client data.}


These pioneering efforts in DQC and DQL illustrate the feasibility and benefits of distributing quantum computation and learning tasks across heterogeneous nodes. However, they also reveal pressing challenges related to convergence behavior in noisy, resource-limited quantum networks and the necessity of post-quantum cryptographic safeguards against quantum-capable adversaries. A key practical challenge in DQL is that aggregating updates from all devices in every round is often inefficient, especially when client models vary significantly in quality. In such settings, system performance depends critically on parameters that must be systematically understood to ensure efficient and accurate learning. At the same time, as DQL systems scale over quantum networks, securing inter-node communication against quantum-capable adversaries becomes essential, requiring the integration of post-quantum cryptographic safeguards alongside quantum networking primitives. In addition to these challenges, portable quantum computers, such as in \cite{xu2021qubic, fruitwala2022distributed, feng2022spinq}, use field programmable gate array (FPGA) based System-on-Chip (SoCs) control boards, often no more powerful than a Raspberry Pi, to handle nearly all operations, including state preparation, control signal generation, and communication. \textcolor{black}{While quantum-native protocols such as quantum key distribution (QKD) offer robust security for model sharing and have achieved experimental milestones, they are far from general implementation for widespread usability. Such systems require dedicated, isolated hardware and precisely engineered control systems to operate due to their noise sensitivity. Thus, the cost of such a system increases substantially, which is justifiable only for high-sensitivity links, such as government, defense, and some financial infrastructure. As a result, the transition from existing security systems to QKD-based systems is expected to be gradual, with adoption concentrated in niche areas. Therefore, public-key cryptographic techniques, namely post-quantum cryptography (PQC), are the more practical approach for near-term implementation. Even a fully mature and scalable QKD system needs a robust authentication solution, which can be achieved via PQC-based signatures \cite{fedorov2023deploying}.} However, using PQC heavily on such devices is resource-intensive, in addition to their quantum processors being inherently noisy. \textit{\textbf{This paper addresses these challenges by presenting a novel integrated framework for convergence analysis under partial participation and a robust post-quantum security mechanism for DQL frameworks.}} In our proposed quantum-secure DQL (QS-DQL) framework, the noise challenge is addressed through fidelity-aware partial client participation, while the resource limitation in classical control boards for securing communication via PQC is mitigated by the dynamic security mechanism, making it the closest framework to practical implementation. 

\subsection{Related Work}
\textit{1) DQC}: Several studies have explored DQC from different perspectives, aiming to overcome the qubit limitations of current NISQ devices. Some works focus on quantum compilation frameworks that account for both network and device constraints, enabling efficient partitioning of quantum algorithms, optimizing remote operation scheduling, and minimizing entanglement consumption~\cite{ferrari2023modular,ferrari2021compiler}. Other studies investigate algorithm distribution strategies, such as parallelizing VQE computations across multiple quantum processors, formulating systematic methods for generating distributed quantum circuits, and proposing centralized or decentralized control architectures~\cite{diadamo2021distributed}. Moreover, research has addressed the qubit allocation problem in DQC, developing heuristic and hybrid optimization algorithms to map logical qubits to physical devices under various network topologies~\cite{mao2023qubit}. Furthermore, fault tolerance and security have been examined through enhanced quantum Byzantine agreement protocols, integrating error mitigation techniques and quantum key distribution-inspired verification mechanisms to improve operational reliability in adversarial settings~\cite{chen2025noise}. \textit{Collectively, these studies highlight the broad range of approaches advancing DQC in scalability, efficiency, and security, yet they do not analyze its convergence behavior and security threats.}

\textit{2) DQL}: Several studies have explored DQL to address scalability, heterogeneity, and privacy challenges in quantum-enhanced machine learning. \textcolor{black}{Authors in a foundational work \cite{chehimi2022quantum} proposed a quantum federated learning framework designed for distributed quantum data, allowing multiple nodes to collaboratively train an unknown quantum state classifier without sharing raw quantum information. Another similar work \cite{chen2021federated} introduced a framework for hybrid quantum-classical networks, which demonstrates that quantum classifiers can be effectively trained on classical image datasets in a distributed setting.} Distributed multi-agent reinforcement learning leverages VQCs to reduce trainable parameters and improve data representation \cite{chen2025quantum}, while auction-based client selection enhances trust and privacy using QNNs \cite{lee2025auction}. Consensus-based distributed quantum kernel learning has been applied to speech recognition for privacy-preserving scalability \cite{chen2025consensus}, and satellite-ground DQL frameworks use slimmable QNNs with superposition coding to improve efficiency \cite{park2024dynamic}. Some work also proposes distributed quantum long short-term memory models that partition VQCs for scalable sequence modeling \cite{chen2025toward}. Further advances include entanglement-controlled DQL for interference mitigation \cite{park2025entanglement}, and dynamic aggregation DQL for robust medical diagnosis under diverse data distributions \cite{qu2025daqfl}. These studies span DQL research in algorithm design, trust management, sequence modeling, communication optimization, and application-specific solutions for scalability, privacy, and robustness. \textit{However, these works largely overlook a unified convergence analysis of DQL under realistic settings that also incorporate adaptive security mechanisms.}

\textit{3) Security in DQC and DQL}: Numerous efforts have been made to secure distributed quantum computation and learning infrastructures. At the hardware level, isolation and scheduling mechanisms have been proposed to mitigate multi-tenant risks like crosstalk leakage \cite{lee2025auction}. Yet, this approach avoids insecure hardware components rather than securing them and requires physical modifications. At the model level, techniques such as random unitary encoders \cite{gong2024enhancing} and differential privacy have been investigated \cite{rofougaran2024federated}, but often at the cost of a decrease in model accuracy and noisier circuits. Regarding the communication layer, while quantum-native protocols like QKD and verifiable blind quantum computing \cite{sheng2017distributed, li2021quantum} offer robust theoretical security, they demand dedicated optical infrastructure and fault-tolerant capabilities that remain impractical for widespread deployment. \textcolor{black}{Recent insights emphasize that the coexistence of QKD and PQC enables scalable quantum infrastructure \cite{fedorov2023deploying}.}  This positions PQC as the essential practical layer for current networks, and conventional PQC has been benchmarked for integrity \cite{gong2024enhancing}. \textit{However, existing work lacks a pragmatic and adaptive framework to handle dynamic network threats.}

\subsection{Key Contributions}
Despite these research efforts, existing works do not address the convergence behavior of DQL systems under partial client participation, nor do they integrate post-quantum cryptography with adaptable security mechanisms and quantified resource overhead. This motivates us to develop a new theoretical DQL framework with a novel post-quantum security design. The main contributions of this paper are as follows: 
\begin{itemize}
    \item We analyze the convergence behavior of DQL under a partial device participation scenario for non-convex loss functions and non-independent and identically distributed (non-IID) data distributions. Furthermore, we derive a novel convergence bound that reveals a fundamental trade-off between convergence rate, the number of measurement shots, and the device subset size (Section \ref{systemmodel}).
    \item We design a multi-layered post-quantum cryptographic security architecture with a real-time dynamic security adaptation mechanism for model sharing powered by a dedicated QNN that monitors channel conditions, detects threat levels, and adjusts PQC settings across three National Institute of Standards and Technology (NIST)-compliant security levels. The proposed dynamic security mechanism is evaluated on real hardware equivalents serving as the classical control structure in quantum computers, profiling the reduction in resource overhead achieved by our approach (Section \ref{sec:security}).
    \item Extensive numerical experiments show that our theoretical findings align with simulation results, confirming the identified trade-off, and further demonstrate the effectiveness of our PQC-based dynamic security mechanism in enhancing protection while maintaining efficient DQL performance (Section \ref{experimentsandresults}).
\end{itemize}

\section{System Model} \label{systemmodel}
QS-DQL system considers a distributed learning framework \textcolor{black}{(Fig.~\ref{fig:overview_Jsac}), where a central quantum server} coordinates with a network of distributed quantum nodes to collaboratively train a global model. Let $\mathcal{U} = \{1, 2, \dots, U\}$, $\mathcal{T} = \{1, 2, \dots, T\}$, and $\mathcal{K} = \{1, 2, \dots, K\}$ denote the sets of devices, local iterations, and global rounds, respectively. In each global round, \textit{a subset of devices} performs local training on their private datasets before communicating their updated model parameters to the server for aggregation. \textcolor{black}{The communication is protected by an adaptive post-quantum cryptographic mechanism that utilizes lattice-based algorithms to resist potential quantum adversaries. A threat detection model on the server dynamically selects the required cryptographic suite, parameter set, and re-keying frequency for that round's transmission based on real-time threat level.}
\begin{figure}
    \centering
    \includegraphics[width=0.90\linewidth]{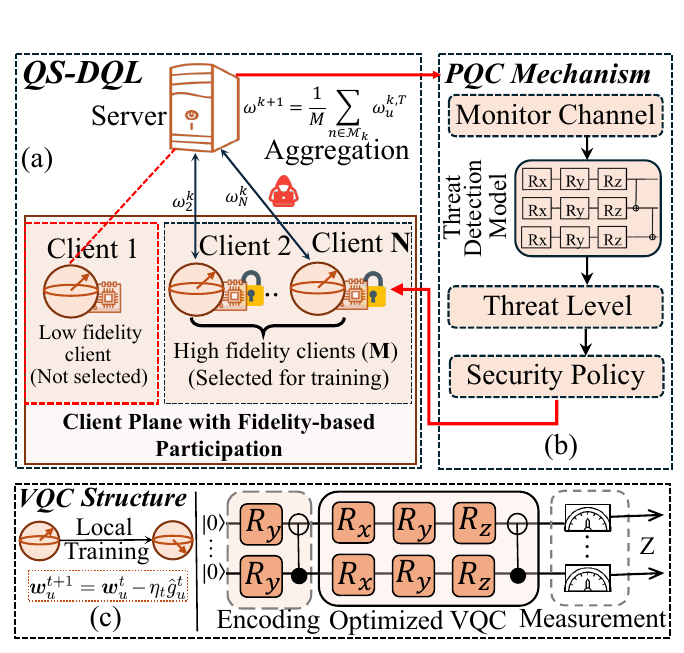}
    \caption{Overview of the QS-DQL framework with (a) fidelity-based partial quantum client participation for collaborative training where $M$ clients have been selected, (b) adaptive post-quantum cryptographic mechanism in the presence of a threat, and (c) local VQC training procedure.}
    \label{fig:overview_Jsac}
\end{figure}

\subsection{Baseline DQL Framework}
The process on each client $u \in \mathcal{U}$ begins with state preparation. Classical input data is first processed by classical neural network layers for dimensionality reduction and then encoded into a $q$-qubit state, initialized in the ground state $\lvert 0 \rangle$. This encoded state is then passed through a VQC, which is the core of the client's QNN.

The VQC applies a parameterized unitary transformation $U(\boldsymbol{w_{u}})$ on the $q$ qubits, where the circuit is parameterized by $\boldsymbol{w_{u}} = (\boldsymbol{w}_{u}^{1}, \dots, \boldsymbol{w}_{u}^{P})$, with $P$ representing the total number of trainable parameters. The specific VQC design, or ansatz, defines the unitary structure as $U(\boldsymbol{w}_{u}) = \prod_{p=1}^{P} U_{p} (\boldsymbol{w}_{u}^{p}) V_{p}$
where each parameter $\boldsymbol{w}_{u}^{p} \in \mathbb{R}$ controls a corresponding unitary gate $U_{p} (\boldsymbol{w}_{u}^{p})$ given by $U_{p} (\boldsymbol{w}_{u}^{p}) = e^{-i\frac{\boldsymbol{w}_{u}^{p}}{2}G_{p}}$,
with $G_p \in \{I, X, Y, Z\}^{\otimes q}$ denoting the Pauli string generator. \textcolor{black}{The unitary $V_{p}$ is a fixed, non-trainable gate.} The VQC produces the final parameterized state for client $u$ as ${\lvert \Psi_{u}(\boldsymbol{w}_{u}) \rangle} = U(\boldsymbol{w}_{u}) \lvert 0 \rangle$. The corresponding pure-state density matrix is $\Psi_{u}(\boldsymbol{w}_{u}) = \lvert \Psi_{u}(\boldsymbol{w}_{u}) \rangle \langle \Psi_{u}(\boldsymbol{w}_{u}) \rvert$.

The objective of the local training is to find the parameter vector $\boldsymbol{w}_{u}$ that minimizes a local loss function, $f_{u}(\boldsymbol{w}_{u})$ as
\begin{align}
    \min_{\boldsymbol{w}_{u} \in \mathbb{R}^{P}} \left\{ f_{u}(\boldsymbol{w}_{u}) := \mathcal{C}(\langle Z \rangle_{\lvert \Psi_{u}(\boldsymbol{w}_{u}) \rangle} , y)\right\}, \label{eqn5}
\end{align}
where $\mathcal{C}(\cdot)$ is a task-specific cost function comparing the model's output with the target label $y$. The model's output is the expectation value of an observable $Z$, computed as $\langle Z \rangle_{\lvert \Psi_{u}(\boldsymbol{w}_{u}) \rangle} = \langle 0 \rvert U^\dagger(\boldsymbol{w}_{u}) Z U(\boldsymbol{w}_{u}) \lvert 0 \rangle = \operatorname{Tr}(Z \Psi_{u}(\boldsymbol{w}_{u}))$. Since $Z$ is Hermitian, it has an eigendecomposition $Z = \sum_{j=1}^{N_z} h_j \Pi_j$, allowing the loss $f_{u}(\boldsymbol{w}_{u})$ to be expressed as
$
     \sum_{j=1}^{N_z} h_j \operatorname{Tr}(\Pi_j \Psi_{u}(\boldsymbol{w}_{u})).
$
During training, this value is estimated empirically from $H$ repeated measurements on the quantum state, yielding samples $\{ Z_{h} \}_{h=1}^{H}$. The empirical estimate of the loss is given by
\begin{align}
    \hat{f}_{u}(\boldsymbol{w}_{u}) = \hat{\langle Z \rangle}_{\lvert \Psi_{u}(\boldsymbol{w}_{u}) \rangle} = \frac{1}{H} \sum_{h=1}^{H} Z_h, \label{eqn10}
\end{align}
where $\hat{(\cdot)}$ denotes the empirical estimate. The local model parameters are then updated using stochastic gradient descent (SGD). For a client $u$ at local iteration $t$, the update rule is expressed as
$
    \boldsymbol{w}_{u}^{t+1} = \boldsymbol{w}_{u}^{t} - \eta_{t} \hat{g}_{u}^{t},
$
with $\eta_{t} > 0$ representing the learning rate and $\hat{g}_{u}^{t}$ being the stochastic gradient estimate.

After each client completes its $T$ local training iterations, the updated local models are transmitted to the central server. The server's goal is to optimize a global loss function, $F(\boldsymbol{w}) = \sum_{u=1}^{U} p_u f_u(\boldsymbol{w})$. To achieve this, the server aggregates the received local models using a federated averaging algorithm to produce the new global model for the next round, $k+1$, given by $\boldsymbol{w}^{k+1} = \frac{1}{M} \sum_{n \in \mathcal{M}_k} \boldsymbol{w}_{u}^{k,T}$. This new global model $\boldsymbol{w}^{k+1}$ is then broadcast to the devices to begin the next global round of training.

\subsection{Fidelity-Aware Client Participation}
In a practical quantum setting, local model updates are not of uniform quality due to hardware noise and calibration drift. This noise primarily manifests as variance in the stochastic gradient estimate, $\Tilde{g}_{u,k}^{t}$. As established in our convergence analysis, this uncertainty is captured by the shot noise variance term, $\text{var}(\xi_{k}^{t})$, which is upper bounded as
\begin{equation}
\text{var}(\xi_{k}^{t}) \leq \frac{1}{U} \sum_{u \in \mathcal{U}} \frac{\nu N_{z} D \operatorname{Tr}(Z^{2})}{2H}. 
\end{equation}
Since devices can exhibit different effective noise levels (captured by parameters like $\nu$) and may use a different number of measurement shots, $H$, their update fidelity varies. To mitigate this, our framework incorporates a fidelity-aware partial-participation strategy, as shown in Algorithm~\ref{algorithm2}. \textcolor{black}{After each local training round $k$, every client $u \in \mathcal{U}$ reports a vector of fidelity metrics, $\boldsymbol{\mathcal{F}}_{u,k} \in \mathbb{R}^D$, which captures the quality and variance of its gradient estimates. This metric is derived directly from the measurement shot budget $H$ used during the parameter-shift rule \cite{wierichs2022general}, utilizing the theoretical property that estimator variance scales inversely with $H$. Expectation values required for the gradients are estimated as empirical means over $H$ repeated shots. Thus, a noise estimate without requiring additional quantum circuit executions just for variance calculation can be obtained. This fidelity metric is a low-dimensional scalar whose transmission overhead is negligible relative to that of the high-dimensional model weights.} The server then applies a selection function, $\text{Select}(\cdot)$, to choose the subset of $M$ devices with the highest fidelity,
$
\mathcal{M}_k = \text{Select}(\{\boldsymbol{\mathcal{F}}_{u,k}\}_{u \in \mathcal{U}}),
$
where $|\mathcal{M}_k| = M$. \textcolor{black}{As this process only involves a standard sorting-based ranking step, it is computationally inexpensive, with a complexity of $O(N \log N)$. Therefore, the additional latency introduced by the server-side selection is minimal. In practice, it is negligible compared to the dominant costs of local training and model communication in typical high-speed network settings.} This dynamic selection process ensures that only the most reliable updates, effectively those with the lowest gradient variance, contribute to the global model.
\begin{algorithm}[ht!]
  \caption{DQL Algorithm with Fidelity-Aware Partial Aggregation}
  \label{algorithm2}
\begin{algorithmic}[1]
 \footnotesize
    \State \textbf{Input:} learning rate $\eta_{k}$, initial global model $\boldsymbol{w}_{0}$, number of local updates $T$, number of selected devices $M$
    \State \textbf{for} $k = 1, 2, \dots, K$ \textbf{do:}
        \State \indent \textbf{Server broadcasts} $\boldsymbol{w}_{k}$ to all devices
        \State \indent \textbf{Parallel for all devices} $u \in \mathcal{U}$ \textbf{do:}
            \State \indent \indent Initialize $\boldsymbol{w}_{u,k}^{0} = \boldsymbol{w}_{k}$
            \State \indent \indent \textbf{for} $t = 0, 1, \dots, T-1$ \textbf{do:}
                \State \indent \indent \indent $\boldsymbol{w}_{u,k}^{t+1} = \boldsymbol{w}_{u,k}^{t} - \eta_{k} \Tilde{g}_{u,k}^{t}$
            \State \indent \indent \textbf{end for}
            \State \indent \indent Each client sends updated model $\boldsymbol{w}_{u,k}^{T}$ to the server
        \State \indent \textbf{end parallel for}
        \State \indent \textbf{Server evaluates fidelity of all received models:}
            \State \indent \indent Compute fidelity scores $\{\mathcal{F}_{u}\}_{u \in \mathcal{U}}$
        \State \indent \textbf{Server applies sorting on fidelity scores and model updates}
            \State \indent \indent Select subset $\mathcal{M} \subset \mathcal{U}$ of $M$ devices with high fidelity
        \State \indent \textbf{Server aggregates:}
            \State \indent \indent $\bar{\boldsymbol{w}}_{k+1} = \frac{1}{M} \sum_{m \in \mathcal{M}} \boldsymbol{w}_{m,k}^{T}$
        \State \indent \textbf{Server broadcasts} $\bar{\boldsymbol{w}}_{k+1}$ to all devices
    \State \textbf{end for}
    \State \textbf{Output:} final global model $\bar{\boldsymbol{w}}_{K}$
\end{algorithmic}
\end{algorithm}

\noindent
\textbf{\textit{Complexity Analysis}:}
We start with the inequality $f_{u}(\boldsymbol{w}) \leq f_{u}(\boldsymbol{w'}) + \nabla {f_{u}(\boldsymbol{w}')}^{T}(\boldsymbol{w}-\boldsymbol{w'}) + \frac{L}{2}||\boldsymbol{w}-\boldsymbol{w'}||^{2}$ along with the gradient bound $||\nabla f_{u}(\boldsymbol{w})||^{2} \geq 2 \mu (f_{u}(\boldsymbol{w}) - {f_{u}}^{*})$. \textcolor{black}{Here, $\mu$ denotes the Polyak-Lojasiewicz (PL) constant (refer to Assumption \ref{Assumption2}). This condition relaxes the strict convexity assumption for non-convex landscapes and characterizes the geometry by ensuring that the gradient magnitude remains proportional to the sub-optimality gap.} It depends on the circuit size, specifically the number of qubits ($n$) and the number of parameters ($P$).

\textcolor{black}{In general, larger circuits tend to produce smaller gradient norms. This effect is due to the barren plateau phenomenon, where gradients tend to vanish exponentially via flattening of the optimization landscape \cite{mcclean2018barren}.} Using standard SGD convergence results from \cite{ajalloeian2020convergence}, for any learning rate $\eta \leq \frac{1}{L}$, the expected loss satisfies $\mathbb{E}[f_{u}(\boldsymbol{w}^{T})] - {L}^{*} \leq (1 - \eta \mu)^{T} ([f_{u}(\boldsymbol{w}^{0})] - {L}^{*}) + \frac{1}{2} [\frac{\eta L V}{\mu}]$, where $V = \frac{\nu N_{z} D Tr(Z^{2})}{2H}$ is an upper bound on the variance of the gradient estimate due to shot noise. To achieve an error of $\mathcal{O}(\delta)$, it is sufficient to select $\eta \leq \min\{\frac{1}{L},\frac{\delta \mu}{L V}\}$ and run $T^{\text{shot-noise}} = \mathcal{O} \bigg( \log\frac{1}{\delta} + \frac{V}{\delta \mu} \bigg) \frac{L}{\mu}$ iterations. 
\textcolor{black}{
This establishes that local model training converges to a persistent error floor of order $\mathcal{O}(\eta L V)$. Crucially, this creates a unique trade-off. Unlike classical systems where variance vanishes with batch size, here $V$ is constrained by the shot count $H$ which is fundamentally quantum. Furthermore, the error floor can be made arbitrarily small by choosing a sufficiently small learning rate.}

\subsection{Convergence Analysis}

We first present \textbf{\textit{notations and definitions}} used for our convergence analysis. We consider the global optimization problem:
\begin{align}
    \min_{\boldsymbol{w}} f(\boldsymbol{w}) \triangleq \sum_{u=1}^{U} f_{u}(\boldsymbol{w}), \label{eqn23}
\end{align}
where $f(\boldsymbol{w})$ is the global objective function. Each client $u$ trains on a local dataset $\mathcal{S}_{u}$ with $S_{u}$ samples from distribution $\mathcal{D}_{u}$. The full local gradient is defined as $g_{u} = \frac{1}{|\mathcal{S}_{u}|} \nabla f_{u}(\boldsymbol{w})\stackrel{\triangle}{=} \frac{1}{|\mathcal{S}_{u}|} \nabla f(\boldsymbol{w};\mathcal{S}_{u})$ and the stochastic gradient as $\Tilde{g_{u}} \stackrel{\triangle}{=} \frac{1}{B} \nabla f(\boldsymbol{w};\xi_{u})$, where $\xi_{u} \subseteq \mathcal{S}_{u}$ is mini-batch of size $B$. At local iteration $t$ of global round $k$, client $u$'s parameters and gradients are $\boldsymbol{w}_{u,k}^{t}$ and $g_{u,k}^{t}$, respectively. We define:
\begin{align}
    \bar{\boldsymbol{w}}_{k}^{t} \stackrel{\triangle}{=} \frac{1}{U} \sum_{u \in \mathcal{U}} \boldsymbol{w}_{u,k}^{t},
    \quad
    \Tilde{g}_{k}^{t} \stackrel{\triangle}{=} \frac{1}{U} \sum_{u \in \mathcal{U}} \Tilde{g}_{u,k}^{t},
\end{align}
The local SGD update is
$
    \boldsymbol{w}_{u,k}^{t+1} = \boldsymbol{w}_{u,k}^{t} - \eta_{k} \Tilde{g}_{u,k}^{t}. \label{eqn30}
$
Since $\mathbb{E} \Tilde{g}_{k}^{t} = g_{k}^{t}$, the global update is an unbiased estimate of the true gradient. Finally, we assume gradient diversity is bounded by $\lambda$, i.e.,
\begin{align}
    \frac{\sum_{u=1}^{U} ||g_{u,k}^{t}||_{2}^{2}}{||\sum_{u=1}^{U}g_{u,k}^{t}||_{2}^{2}} \leq \lambda. \label{eqn31}
\end{align}
Next, we outline the key \textbf{\textit{assumptions}} that form the basis of our convergence analysis.
\begin{assumption}[Smoothness and Lower Boundedness] \label{Assumption1}
\textit{The local objective function $f_{u}(.)$ associated with device $u$ is differentiable for $1 \leq u \leq U$ and is $L-smooth$, i.e., $||\nabla f_{u}(\mathbf{u}) - \nabla f_{u}(\mathbf{v})|| \leq L||\mathbf{u}-\mathbf{v}||, \forall{\mathbf{u},\mathbf{v}} \in \mathbb{R}^{d}$.}
\end{assumption}
\begin{assumption}[$\mu$-Polyak-Lojasiewicz (PL)] \label{Assumption2}
\textit{The global objective function $f(.)$ is differentiable and satisfy the Polyak-Lojasiewicz (PL) condition with constant $\mu$, i.e., $\frac{1}{2} ||\nabla f(\boldsymbol{w})||_{2}^{2} \geq \mu (f(\boldsymbol{w})-f({\boldsymbol{w}}^{*}))$ holds $\forall \boldsymbol{w} \in \mathbb{R}^{d}$ with ${\boldsymbol{w}}^{*}$ being the optimal solution of global objective.}
\end{assumption}
\begin{assumption}[Bounded Local Variance] \label{Assumption3}
\textit{For every local dataset $S_{u}$, $u = 1, 2, \dots, U$, we can sample an independent mini-batch $\xi_{u} \subseteq \mathcal{S}_{u}$ with $|\xi_{u}| = B$ and compute an unbiased stochastic gradient $\Tilde{g_{u}} = \frac{1}{B} \nabla f(\boldsymbol{w};\xi_{u})$, $\mathbb{E}[\Tilde{g_{u}}] = g_{u} = \frac{1}{|\mathcal{S}_{u}|} \nabla f(\boldsymbol{w};\mathcal{S}_{u})$ with the variance bounded as \cite{jose2022error}
\begin{align}
    \mathbb{E}[||\Tilde{g_{u}} - g_{u}||^{2}] \leq C_{1} ||g_{u}||^{2} + \frac{\sigma^{2}}{B} + \text{var}(\xi_{k}^{t}). 
\end{align}
where $C_1$ is a non-negative constant and inversely proportional to the mini-batch size, $\sigma$ is another constant controlling the variance bound, and $\text{var}(\xi_{k}^{t})$ is the variance of the gradient estimate.}
\end{assumption}
The left-hand side captures the total uncertainty in the stochastic gradient, stemming from (i) mini-batch sampling variance and (ii) shot noise variance, $\text{var}(\xi_{k}^{t})$, where
$
    \Tilde{g}_{u,k}^{t} = g_{u,k}^{t} + \xi_{u,k}^{t}. \label{eqn10'''}
$
Here, $\Tilde{g}_{u,k}^{t}$ is the estimated gradient, $g_{u,k}^{t}$ the true gradient, and $\xi_{u,k}^{t}$ the shot-induced error, capturing both classical non-IID and measurement uncertainties. 
\noindent
Using Assumption \ref{Assumption1} and the update rule in \eqref{eqn30}, we obtain
\begin{align}
    f(\bar{\boldsymbol{w}}_{k}^{t+1}) - f(\bar{\boldsymbol{w}}_{k}^{t}) \leq - \eta_{k} \langle \nabla f(\bar{\boldsymbol{w}}_{k}^{t}), \Tilde{g}_{k}^{t} \rangle + \frac{\eta_{k}^{2}  L}{2} ||\Tilde{g}_{k}^{t}||^{2}. \label{eqn37new}
\end{align}
Taking expectations on both sides gives
\begin{align}
    \mathbb{E}[f(\bar{\boldsymbol{w}}_{k}^{t+1}) - f(\bar{\boldsymbol{w}}_{k}^{t})] \leq -\eta_{k} \mathbb{E}[\langle \nabla f(\bar{\boldsymbol{w}}_{k}^{t}), \Tilde{g}_{k}^{t} \rangle] + \frac{\eta_{k}^{2}  L}{2} \mathbb{E}[||\Tilde{g}_{k}^{t}||^{2}]
\end{align}
Averaging over all global rounds and local iterations yields
\begin{align}
    &\frac{1}{KT} \sum_{k=1}^{K} \sum_{t=1}^{T} \mathbb{E}[f(\bar{\boldsymbol{w}}_{k}^{t+1}) - f(\bar{\boldsymbol{w}}_{k}^{t})] \nonumber \\
    & \hspace{6em}\leq \frac{1}{KT} \sum_{k=1}^{K} \sum_{t=1}^{T} (-\eta_{k} \mathbb{E}[\langle \nabla f(\bar{\boldsymbol{w}}_{k}^{t}), \Tilde{g}_{k}^{t} \rangle]) \nonumber \\
    & \hspace{6em}+ \frac{1}{KT} \sum_{k=1}^{K} \sum_{t=1}^{T} \frac{\eta_{k}^{2}  L}{2} \mathbb{E}[||\Tilde{g}_{k}^{t}||^{2}]. \label{eqn35}
\end{align}
We derive bounds for the right-hand-side of \eqref{eqn35}: Lemmas~\ref{lemma1}–\ref{lemma2} bound the first term, and Lemma~\ref{lemma3} bounds the second, together fully characterizing \eqref{eqn35}.

\textit{We implement partial client participation via fidelity-aware aggregation}; in each round, the server measures the fidelity of all received models, groups similar updates and selects a subset $\mathcal{M}$ ($|\mathcal{M}| = M \leq U$) with the highest fidelities, and aggregates them via federated averaging to update the global model until convergence, addressing \eqref{eqn23}. Specifically, at the end of $T$ local iterations, the server selects a subset $\mathcal{M}$ of devices for aggregation. Here, $g_{u,k}^{t} = \frac{1}{|\mathcal{S}_{u}|} \nabla f_{u}(\boldsymbol{w}_{u,k}^{t})\stackrel{\triangle}{=} \frac{1}{|\mathcal{S}_{u}|} \nabla f(\boldsymbol{w}_{u,k}^{t};\mathcal{S}_{u})$ and $\Tilde{g}_{u,k}^{t} \stackrel{\triangle}{=} \frac{1}{B} \nabla f(\boldsymbol{w}_{u,k}^{t};\xi_{u})$. Hence, $\bar{\mathbf{v}}_{k}^{t+1} = \bar{\boldsymbol{w}}_{k}^{t} - \eta_{k} \Tilde{g}_{k}^{t}$ and $\mathbb{E} \Tilde{g}_{k}^{t} = g_{k}^{t}$. The update rule for partial participation is
\begin{align}
     \mathbf{v}_{u,k}^{t+1} = \boldsymbol{w}_{u,k}^{t} - \eta_{k} \Tilde{g}_{u,k}^{t},
\end{align}
where $\boldsymbol{w}_{u,k}^{t+1} =  \mathbf{v}_{u,k}^{t+1} \text{for any local round } t$, and $\text{samples } \mathcal{M} \text{ and averages } \{\mathbf{v}_{u,k}^{T}\}_{u \in \mathcal{M}} \text{for any global round } k$.

For partial client participation, $\bar{\boldsymbol{w}}_{k}^{t+1} = \bar{\mathbf{v}}_{k}^{t+1}$ holds in expectation, requiring an unbiased sampling and aggregation scheme. 
\begin{align}
    \mathbb{E}_{\mathcal{M}} \bar{\boldsymbol{w}}_{k+1} = \bar{\mathbf{v}}_{k+1}, \qquad \bar{\boldsymbol{w}}_{k+1} = \frac{1}{M} \sum_{u \in \mathcal{M}} \mathbf{v}_{u,k+1}.
\end{align}
Here, $\mathcal{M}$ is sampled with replacement using probabilities $q_{1}, q_{2}, \dots, q_U$. Lemma \ref{lemma5} shows unbiasedness, and Lemma \ref{lemma6} bounds the variance of $\bar{\boldsymbol{w}}_{u,k}^{t}$.

\subsection{Lemmas}
The following lemmas form the basis for establishing the main convergence rate.
\begin{lemma} \label{lemma1}
Let Assumption \ref{Assumption1} hold, the expected inner product between stochastic gradient and full gradient is bounded by
\begin{align}
    & -\eta_{k} \mathbb{E}\bigg[\langle \nabla f(\bar{\boldsymbol{w}}_{k}^{t}), \Tilde{g}_{k}^{t} \rangle\bigg] \leq -\frac{\eta_{k}}{2}||\nabla f(\bar{\boldsymbol{w}}_{k}^{t})||^2 \nonumber \\
    & - \frac{\eta_{k}}{2}||\sum_{u=1}^{U}\nabla f_{n}(\boldsymbol{w}_{u,k}^{t})||^{2} + \frac{\eta_{k} L^{2}}{2} \sum_{u=1}^{U} ||\bar{\boldsymbol{w}}_{k}^{t} - \boldsymbol{w}_{u,k}^{t}||^{2}.
\end{align}
\end{lemma}
\begin{proof}
See Section  \ref{prooflemma1} in Appendix. \renewcommand{\qedsymbol}{}
\end{proof}

\begin{lemma} \label{lemma2}
For any global averaging round $k$
\begin{align}
    \mathbb{E}_{\mathcal{M}} (\bar{\boldsymbol{w}}_{k+1}) = \bar{\mathbf{v}}_{k+1}. \label{eqn70}
\end{align}
\end{lemma}
\begin{proof}
See Section \ref{prooflemma2} in Appendix. \renewcommand{\qedsymbol}{}
\end{proof}

\begin{lemma} \label{lemma3}
For any global averaging round $k$, we assume that $\eta_{k}$ is non-increasing and $\eta_{k} \leq 2 \eta_{k+1}$ for all $t \geq 0$. The expected difference between $\bar{\mathbf{v}}_{k+1}$ and $\bar{\boldsymbol{w}}_{k+1}$ is bounded by
\begin{align}
    \mathbb{E}_{\mathcal{M}} || \bar{\mathbf{v}}_{k+1} - \bar{\boldsymbol{w}}_{k+1} ||^{2} \leq \frac{4}{M} \eta_{k}^{2} T^{2} G^{2}.
\end{align}
\end{lemma}
\begin{proof}
See Section \ref{prooflemma3} in Appendix. \renewcommand{\qedsymbol}{}
\end{proof}

\begin{lemma} \label{lemma4}
Under Assumption \ref{Assumption3}, the expected upper bound of $\mathbb{E}[||\Tilde{g}_{k}^{t}||^2]$ is expressed as
\begin{align}
    &\mathbb{E}\bigg[||\Tilde{g}_{k}^{t}||^2\bigg] \leq \lambda \bigg(\frac{C_{1}}{U}+1\bigg) \bigg[\sum_{u=1}^{U}||\nabla f_{u}(\boldsymbol{w}_{u,k}^{t})||^{2}\bigg] + \frac{\sigma^{2}}{UB} \nonumber \\
    & \hspace{15em} + \frac{\text{var}(\xi_{u,k}^{t})}{U}.
\end{align}
\end{lemma}
\begin{proof}
See Section \ref{prooflemma4} in Appendix. \renewcommand{\qedsymbol}{}
\end{proof}

\begin{lemma} \label{lemma5}
The variance of the gradient estimate $\Tilde{g}_{u,k}^{t} = g_{u,k}^{t} + \xi_{u,k}^{t}$ can be upper bounded as 
\begin{align}
\text{var}(\xi_{k}^{t}) \leq \frac{1}{U} \sum_{u \in \mathcal{U}} \frac{\nu N_{z} D Tr(Z^{2})}{2H}. \label{eqn45}
\end{align}
\end{lemma}
\begin{proof}
See Section \ref{prooflemma5} in Appendix. \renewcommand{\qedsymbol}{}
\end{proof}

\begin{lemma} \label{lemma6}
Provided that Assumption \ref{Assumption3} is fulfilled, the expected upper bound of the divergence of $\boldsymbol{w}_{u,k}^{t}$ is given as 
\begin{align}
     &\frac{1}{KT} \sum_{k=1}^{K} \sum_{t=1}^{T} \sum_{u=1}^{U} \bigg[\mathbb{E} ||\bar{\boldsymbol{w}}_{k}^{t} - \boldsymbol{w}_{u,k}^{t}||^{2}\bigg] \leq \frac{4}{M} \eta_{k}^{2} T^{2} G^{2} \nonumber \\
     &+ \frac{(2 C_{1} + T(T+1))}{KT} \frac{\lambda \eta_{k}^{2} (U+1)}{U} \sum_{k=1}^{K} \sum_{t=1}^{T} || \sum_{u=1}^{U} g_{u,k}^{t}||^{2} \nonumber \\
     &+ \frac{\eta_{k}^{2} KT(U+1)(T+1)\sigma^{2}}{UB}. \label{eqn78}
\end{align}
\end{lemma}
\noindent
\begin{proof}
See Section \ref{prooflemma6} in Appendix. \renewcommand{\qedsymbol}{}
\end{proof}

\noindent
\textcolor{black}{\textit{\textbf{Theorem 1.}}}
\textit{Let Assumptions \ref{Assumption1}, \ref{Assumption2}, \ref{Assumption3} hold, the convergence rate of the global model training with partial client participation is upper bounded after $K$ global rounds as}

{\footnotesize
\begin{align}
        &\frac{1}{KT} \sum_{k=0}^{K} \sum_{t=0}^{T} \mathbb{E}||\nabla f(\bar{\boldsymbol{w}}_{k}^{t})||^{2} \leq \frac{2 [f(\bar{\boldsymbol{w}}_{1}^{0}) - f^{*}]}{\eta_{k} KT} + \frac{L \eta \sigma^{2}}{UB} \nonumber \\
        &+ \frac{2 \eta_{k}^{2} \sigma^{2} L^{2} (T+1)}{B} \bigg(1+\frac{1}{U}\bigg) + \frac{1}{U} \sum_{u \in \mathcal{U}} \frac{\nu N_{z} D Tr(Z^{2})}{2H} + C, \label{eqn79new}
\end{align}}
\textcolor{black}{\textit{where $C$ denotes the additional penalty induced by partial client participation (i.e., aggregating only $M$ selected clients per round). This term captures the mismatch between the ideal full-participation averaging behavior and the practical partial-aggregation update. In particular, $C$ scales as $C \propto \frac{1}{M} T^{2} \eta_{k_{0}}^{2} G^{2}$. Therefore, in the worst-case upper bound, the dominant participation-related factors are the subset size $M$ (inverse effect) and the local update interval $T$ (quadratic effect), while larger step sizes and larger gradient magnitudes further amplify this penalty.}}
\begin{proof}
See Section \ref{prooftheorem1} in the Appendix. \renewcommand{\qedsymbol}{}
\end{proof}

\begin{remark}
Theorem 1 reveals that DQL convergence under partial participation is influenced not only by the total iterations, total devices, and number of measurements, but also significantly by the size of the client subset chosen in each round.
\end{remark}

\begin{remark}
From \eqref{eqn79new}, larger subsets and increased measurement shots enhance convergence, since averaging across more devices mitigates noise accumulation. This trend is further validated through the numerical experiments in Experiments and Results.
\end{remark}

\section{Design of Adaptive Post-Quantum Security in DQL}\label{sec:security}
A core section of the proposed QS-DQL system is a real-time, extremely practical, and adaptive security framework that is designed to provide robust, quantum-resistant protection while intelligently enhancing the resource efficiency in a quantum node without a trade-off in security.
\subsection{Framework Overview}
Our post-quantum security mechanism is based on the concept of dynamic versatility that moves beyond static security measures often found to be inefficient or insufficient. In a high-security circumstance, enforcing a constant level of maximum applicable security seems a viable option. However, given the dynamic nature of threats in practical scenarios \cite{li2025threats}, our system adjusts its protective measures in real-time based on the perceived threat environment, ultimately avoiding the overhead of naively applying maximum post-quantum cryptographic measures. This is achieved through a continuous and server-orchestrated closed-loop control mechanism, operating over a wireless TCP/IP local area network, which can be segregated as follows
\begin{enumerate}
\item \textit{Monitor:} The central server observes all incoming network traffic associated with the DQL process.
\item \textit{Detect \& Classify:} A trained QNN model analyzes the traffic features to detect and classify potential threats into specific attack categories.
\item \textit{Adapt:} Based on the classification, the server selects a corresponding security policy for the next global training epoch. This policy dictates a holistic set of parameters across multiple system layers.
\end{enumerate}
\textcolor{black}{
The adaptive response is a layered mechanism coordinating across three distinct operational domains. The foundational cryptographic layer modifies the underlying strength of the PQC primitives (for example, changing the level of key encapsulation from lower to a higher NIST security level). Building on this, the protocol layer controls the frequency of operations, like increasing the digital signing rate. The control in the application layer directly influences the DQL process by adjusting the device participation rate to manage the performance impact. }

\subsection{Threat Model and Detection} \label{subsec:threat model}
The QS-DQL framework considers a series of threats that are found to be prominent in different levels of a distributed training network. One is a threat from eavesdroppers who can intercept the communications, posing a significant threat. This evolves into “Harvest Now, Decrypt Later” attacks in the post-quantum era \cite{singh2024managing}, where the encrypted data is recorded today for decryption once powerful quantum computers emerge. Even the lightweight post-quantum cryptographic techniques remain at risk from ML-based and side-channel cryptanalysis methods \cite{li2023salsa}. 
Beyond passive eavesdropping, the threat model includes active attacks available in the \textcolor{black}{UNSW-NB15 dataset \cite{moustafa2015unsw}}, classified by intent and severity, mapped into three threat levels to trigger adaptive security responses.

 \noindent\textit{Level 0: No Threat.} This is the baseline operational state, corresponding to traffic classified as 'Normal'. The system assumes a secure network, optimizing performance with minimal security while retaining baseline protection.
 
\noindent\textit{Level 1: Reconnaissance Threat.} This level considers threats at the communication layer. It is triggered when the server's threat detection model analyzes traffic from external sources and identifies patterns matching 'Reconnaissance' or 'Analysis' attacks.
 
 \noindent\textit{Level 2: Active Intrusion Threat.} This is the highest threat state, triggered when the system detects a direct attempt to compromise its components. This level corresponds to the server analyzing traffic originating from a known client and classifying it as an 'Exploit' or 'Generic' attack, indicating a potential client compromise or other severe threats.

\textcolor{black}{To formalize the detection of these threats, we deploy a trained QNN model, $f_{det}$, on the central server. The server actively monitors every incoming network flow to distinguish between external threats from unknown parties and internal threats where anomalous traffic originates from a known client. The input to this model is a 32-feature vector, $X_k \in \mathbb{R}^{32}$. This vector is generated by simulating network traffic crafted to replicate the statistical features of attack vectors from the training dataset. For each global training round, the input vector $V_n \in \mathbb{R}^d$ is formed by measuring 7 dynamically generated traffic features (\texttt{dur}, \texttt{sbytes}, \texttt{spkts}, \texttt{sload}, \texttt{smeansz}, \texttt{sintpkt}, \texttt{sjit}) and appending the remaining 25 features from a corresponding sample in the training data.} The classifier maps this feature vector to a predicted attack class, $c_t$, from the set of possible classes $\mathcal{Y}$, expressed as $
    c_t = f_{det}(V_t), c_t \in \mathcal{Y}.
$
The predicted class $c_t$ is then mapped to the corresponding discrete system security level, $L_t \in \{0, 1, 2\}$, by a deterministic policy function $g_{det}(\cdot)$ that implements the mapping defined above. Mathematically, it is expressed as
$
    L_t = g_{det}(c_t)
$.
The choice of a QNN is strategic, supported by a growing body of research highlighting its suitability for intrusion detection \cite{subramanian2024hybrid, kukliansky2024network}, along with architectural consistency with DQL. The inherent quantum-based model with obliteration of a large parameter count in classical models makes QNN further suitable for real-time threat detection with reduced training and inference time \cite{abbas2021power}.
\subsection{Post-Quantum Cryptography Primitives}
To provide quantum-resistant security, our framework employs the \textcolor{black}{CRYSTALS suite of cryptographic algorithms \cite{aikata2022kali}}, standardized by the NIST.

\noindent\textit{CRYSTALS-Kyber:} For confidentiality, we use CRYSTALS-Kyber, a Key Encapsulation Mechanism (KEM) based on the Module-LWE problem. It offers multiple versions (e.g., ML-KEM-512, 768, 1024) corresponding to NIST security levels 1, 3, and 5. The role of this PQC algorithm is to securely establish a shared secret, from which a 256-bit key is derived. This key is then used for the actual encryption of all model updates using AES-GCM, an efficient Authenticated Encryption with Associated Data (AEAD) scheme that pairs it with a unique 96-bit nonce for each message.

\noindent\textit{CRYSTALS-Dilithium:} For authenticity and integrity, we use CRYSTALS-Dilithium, a Digital Signature Algorithm (DSA) also based on hard lattice problems. It provides multiple versions (e.g., ML-DSA-44, 65, 87) for NIST security levels 2, 3, and 5, and allows devices to sign model updates, enabling the server to verify data origin and integrity.

\subsection{Adaptive Security Levels}
For each security level $L_t \in \{0, 1, 2\}$ determined by the threat detection model, the framework enforces a specific security policy vector, $\mathbf{P}_L$. This vector defines the cryptographic primitives and application parameters for the duration of that level. In such a granular and multi-layered response, the policy vector is formally defined as
\begin{equation}
    \mathbf{P}_L = (\text{KEM}_L, \text{DSA}_L, \sigma_L, \mathcal{K}_{re, L}, m_L)
\end{equation}
where $\text{KEM}_L$ is the ML-KEM variant from CRYSTALS-Kyber, 
$\text{DSA}_L$ is the ML-DSA variant from CRYSTALS-Dilithium, 
$\sigma_L$ is the fraction of model updates digitally signed, $\mathcal{K}_{re,L}$ is the re-keying frequency, 
and $m_L$ is the fraction of $U$ devices selected per global round.

The specific policies for each level are detailed in Table~\ref{tab:adaptive_security_levels}.

\begin{table}[htbp]
\centering
\caption{Adaptive Security Policy Vector Definitions}
\label{tab:adaptive_security_levels}
\resizebox{\columnwidth}{!}{
\begin{tabular}{lccc}
\toprule
\textbf{Parameter/Threat Level} & \textbf{Level 0} & \textbf{Level 1} & \textbf{Level 2} \\
\midrule

\textbf{Kyber Version} ($\text{KEM}_L$) & ML-KEM-512 & ML-KEM-768 & ML-KEM-1024 \\
\textbf{Dilithium Ver.} ($\text{DSA}_L$) & ML-DSA-44 & ML-DSA-65 & ML-DSA-87 \\
\textbf{Signing Rate} ($\sigma_L$) & Init only & 0.2 (Every 5 rounds) & 1.0 (Always On) \\
\textbf{Re-keying} ($\mathcal{K}_{re}$) & Init only & 0.2 (Every 5 rounds) & 1.0 (Always On)\\
\textbf{Client Participation} ($m_L$) & 0.8 (Optimal) & 0.8 (Optimal) & 0.5 (Reduced) \\
\bottomrule
\end{tabular}
}
\end{table}

\noindent\textbf{Security Level 0:} A default operational state, optimized for performance in a trusted environment. Here, $\mathbf{P}_0$ uses the most lightweight configuration to minimize overhead, where it operates at the optimal client participation rate of 80\% ($m_0=0.8$).

\noindent\textbf{Security Level 1:} This level is activated in response to probing threats. Policy $\mathbf{P}_1$ escalates the cryptographic strength to NIST Level 3 primitives (ML-KEM-768, ML-DSA-65) and introduces periodic signing every 5 global rounds ($\sigma_1=0.5$) for integrity monitoring. Crucially, the framework is designed to absorb this additional security cost while maintaining the optimal 80\% client participation rate ($m_1=0.8$).

\noindent\textbf{Security Level 2:} This is a high-security lockdown state, activated in response to direct attacks. The corresponding policy ($\mathbf{P}_2$) is enforced, which includes excluding the client identified as the source of the threat and reducing the total device participation to 50\% ($m_2=0.5$). This policy also mandates the highest-strength PQC primitives (NIST Level 5), the signing of all training messages ($\sigma_2=1.0$), and an increased re-keying frequency.

\begin{algorithm}[!ht]
\footnotesize
\caption{Real-time Adaptive Post-Quantum Security}
\label{algo:adaptive_pqc_security}
\begin{algorithmic}[1]
\State \textbf{Input:} Threat detector $f_{\text{det}}$, mapping $g_{\text{det}}$, policy set $\{\mathbf{P}_L\}_{L \in \{0,1,2\}}$
\State \textbf{Output:} Security policy $\mathbf{P}_{L_t}$ for round $t+1$
\For{each round $t$}
    \State Extract traffic features $V_t$
    \State Predict class $c_t \gets f_{\text{det}}(V_t)$
    \State Map to level $L_t \gets g_{\text{det}}(c_t)$
    \State Apply policy $\mathbf{P}_{L_t} = (\text{KEM}_{L_t}, \text{DSA}_{L_t}, \sigma_{L_t}, \mathcal{K}_{re,L_t}, m_{L_t})$
    \State Configure PQC parameters, signing rate, re-keying, and client participation
    \State Distribute settings to devices for next round
\EndFor
\end{algorithmic}
\end{algorithm}
\textcolor{black}{
\noindent Algorithm \ref{algo:adaptive_pqc_security} summarizes the adaptive PQC mechanism in QS-DQL. At each round, the server extracts traffic features (line 4), classifies the threat level (line 5), and maps it to a security policy (line 6). The corresponding PQC settings are then enforced and distributed to devices for the next round (lines 7–9).}

\section{Experiments and Results} \label{experimentsandresults}
\subsection{Experimental Setup}
\noindent\textbf{Hardware Testbed:}
Practical quantum computing systems in the near-term operate in a hybrid model, where classical processors concurrently handle essential tasks such as state preparation, gate control signals, VQC parameter optimization, and all communication protocols. Inspired by this architecture, as seen in systems like SpinQ Gemini \cite{feng2022spinq}, which use FPGAs as classical control units, our experimental testbed (Fig. \ref{fig:hardware_testbed}) models these classical components using a cluster of Raspberry Pi devices. This approach allows us to precisely analyze the resource stress imposed by PQC on the control systems, with the understanding that experiments on physical quantum processors would introduce additional hardware noise affecting model convergence.

\begin{figure}
    \centering
    \includegraphics[width=0.8\linewidth]{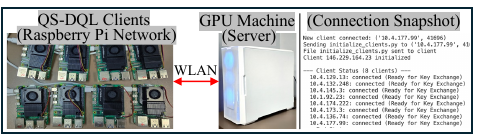}
    \caption{Hardware testbed setup for QS-DQL experiments. A Raspberry Pi network (clients) is connected to a GPU machine (server) via WLAN with a snapshot of the initialized connections.}
    \label{fig:hardware_testbed}
\end{figure}

Our hardware consists of a central server equipped with an NVIDIA GeForce RTX 4090 GPU and eight Raspberry Pi 5 devices. The server acts as the aggregator, while the Raspberry Pis simulate the distributed quantum nodes performing VQC training using TorchQuantum \cite{wang2022quantumnas}, a PyTorch-based state vector simulator. \textcolor{black}{Although this library uses a state-vector backend, it allows hardware compatibility by involving operations with standard native gates (e.g., rotation and CNOTs) and providing built-in conversion tools for Qiskit and OpenQASM deployments.} Initial experiments to analyze baseline DQL convergence were simulated on the server, while the full QS-DQL framework was deployed on the physical Raspberry Pi network. This setup allows us to analyze the actual resource stress imposed by PQC on the classical control systems within a DQL environment.
\\

\noindent\textbf{Datasets and Parameter Settings:} To evaluate the performance of the proposed frameworks, we consider a diverse set of datasets. The initial evaluations, including baseline DQL and its performance on fidelity-based partial participation, use CIFAR-10 \cite{krizhevsky2009learning} and MNIST \cite{deng2012mnist}. CIFAR-10 contains 60,000 color images (32×32 pixels) evenly split into 10 classes, while MNIST consists of 70,000 grayscale images (28×28 pixels) of handwritten digits in 10 classes. For the security extension in QS-DQL, we incorporate the OrganAMNIST dataset from the MedMNIST collection \cite{yang2023medmnist} to simulate training with sensitive medical data. OrganAMNIST includes 58,850 grayscale abdominal CT images (28×28 pixels) from 11 organ classes, representing privacy-critical healthcare data that demands higher security measures. Across all datasets, input images are passed through a single deep layer for dimensionality reduction, producing feature vectors matched to 8–10 qubits using amplitude encoding for state preparation. The subsequent DQL training for all experiments was conducted for 50 global rounds, with each client performing 5 local epochs per round using an SGD optimizer and a learning rate of 0.001. The experiments for evaluating the convergence behavior of DQL consider partial device participation, with 5 to 15 devices per communication round, under both IID and non-IID data distributions. \textcolor{black}{To investigate the impact of measurement shots on convergence, we vary the shot counts (MNIST: 1, 40, 100; CIFAR-10: 1, 100, 300). We include the single-shot ($H=1$) baseline to test the algorithm's robustness under maximum noise, while the higher values allow us to identify the point of diminishing returns. We selected this range to isolate the non-linear impact of noise, avoiding the variance saturation that may potentially mask these trends at higher shot counts.}

\textcolor{black}{For the threat detection QNN model, we utilize the UNSW-NB15 network intrusion dataset \cite{moustafa2015unsw}, specifically isolating four high-impact attack categories (Reconnaissance, Analysis, Generic, and Exploit) against normal traffic. The model architecture is constructed on a 5-qubit quantum circuit structure ($n=5$). To enable the processing of high-dimensional data (32 feature vectors including key features such as \texttt{dur}, \texttt{sbytes}, \texttt{spkts}, \texttt{sload}, \texttt{smeansz}, \texttt{sintpkt}, \texttt{sjit}), we employ amplitude encoding. This encoding method can map $2^n = 32$ selected classical features into the amplitudes of the quantum state vector. Before encoding, these feature vectors are $L_2$-normalized to satisfy the normalization constraint $\sum |x_i|^2 = 1$ required for unitary state preparation. Following the feature map, the state evolves through a variational circuit composed of 3 layers on the 5-qubit structure. Each layer consists of parameterized rotation gates ($R_y$ and $R_z$) applied to every qubit, followed by a sequence of circular CNOT entanglers designed to capture non-linear feature correlations. The measurement is done on Pauli-Z basis, and the output is fed into a softmax layer, whose dimension is equal to the number of classes. The network is trained as a multi-class classifier to categorize network traffic into three distinct attack levels. Training is performed for 100 epochs using the SGD optimizer with a learning rate of 0.001. }To simulate a realistic compromised environment, the security breach is simulated in a single client's communication channel with high fidelity and is constantly selected in the fidelity-based client selection procedure.

\subsection{Convergence Behavior of Baseline DQL Systems}
\begin{figure}[ht!]
    \centering
    \footnotesize
    \begin{subfigure}[t]{0.49\linewidth} 
        \centering
        \includegraphics[width=\linewidth]{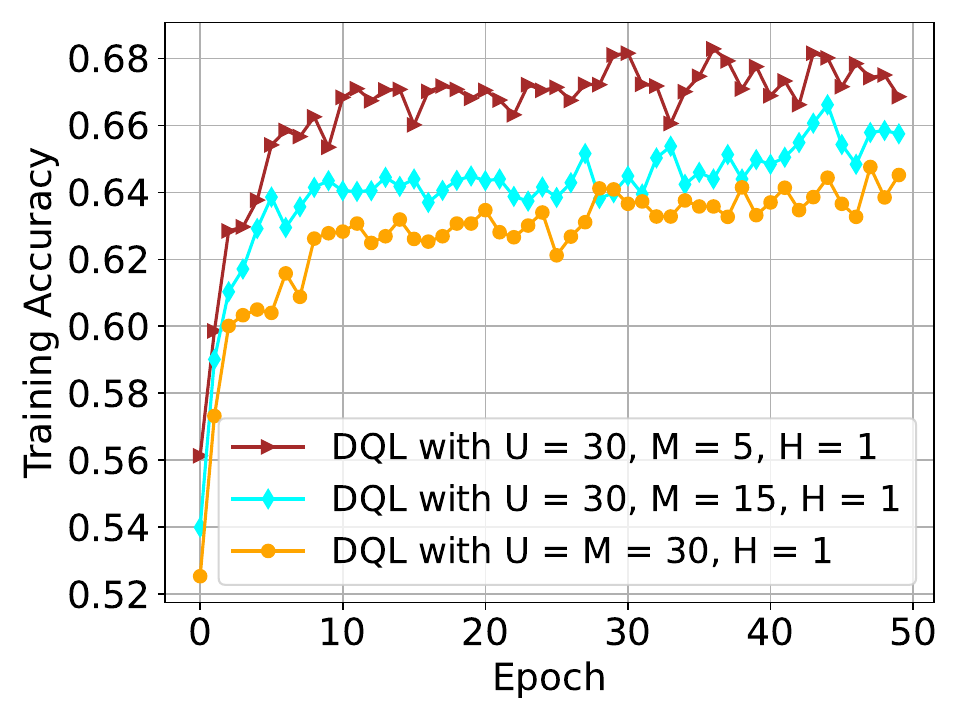}
        \caption{\footnotesize Training accuracy for single measurement shot.}
        \label{fig1a}
    \end{subfigure}
    \hfill 
    \begin{subfigure}[t]{0.49\linewidth} 
        \centering
        \includegraphics[width=\linewidth]{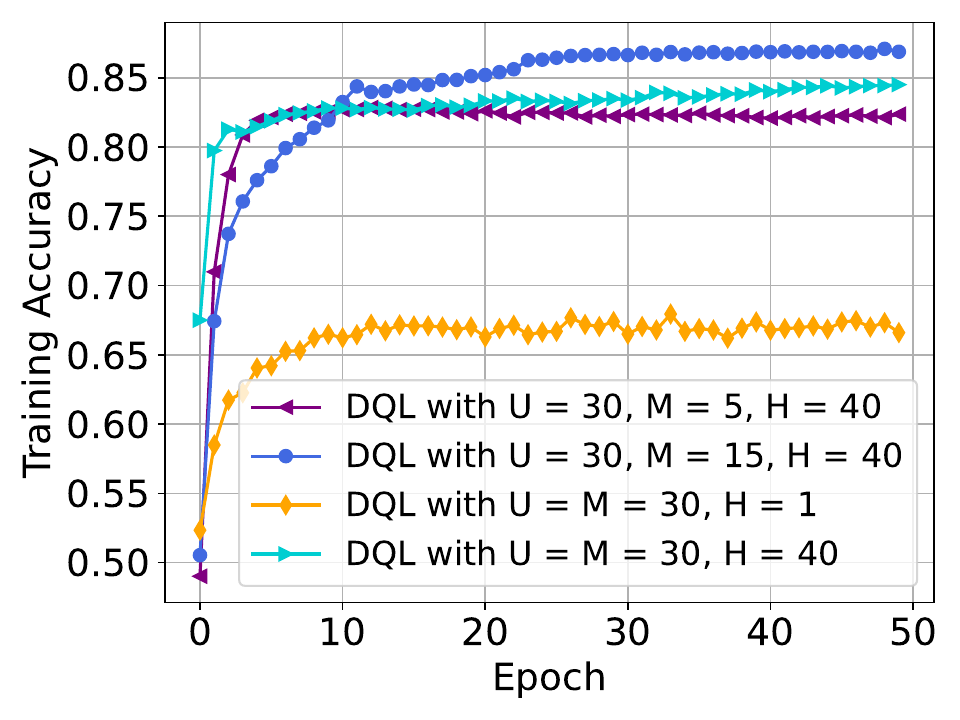}
        \caption{\footnotesize Training accuracy for higher measurement shots.}
        \label{fig1b}
    \end{subfigure}
    \vspace{5pt} 
    \begin{subfigure}[t]{0.49\linewidth} 
        \centering
        \includegraphics[width=\linewidth]{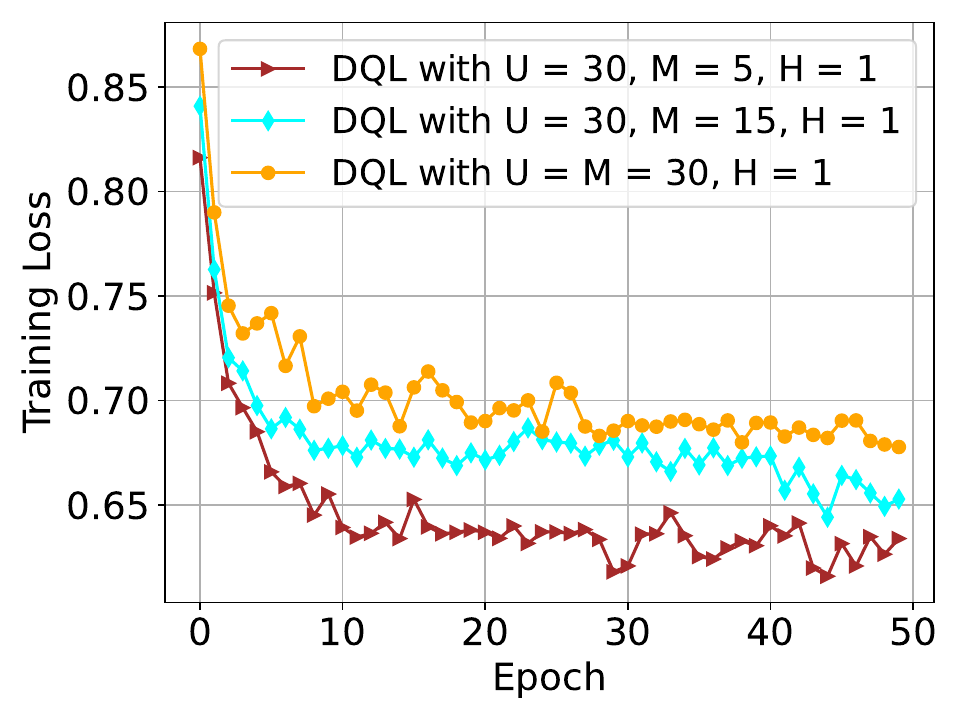}
        \caption{\footnotesize Training loss for single measurement shot.}
        \label{fig1c}
    \end{subfigure}
    \hfill 
    \begin{subfigure}[t]{0.49\linewidth} 
        \centering
        \includegraphics[width=\linewidth]{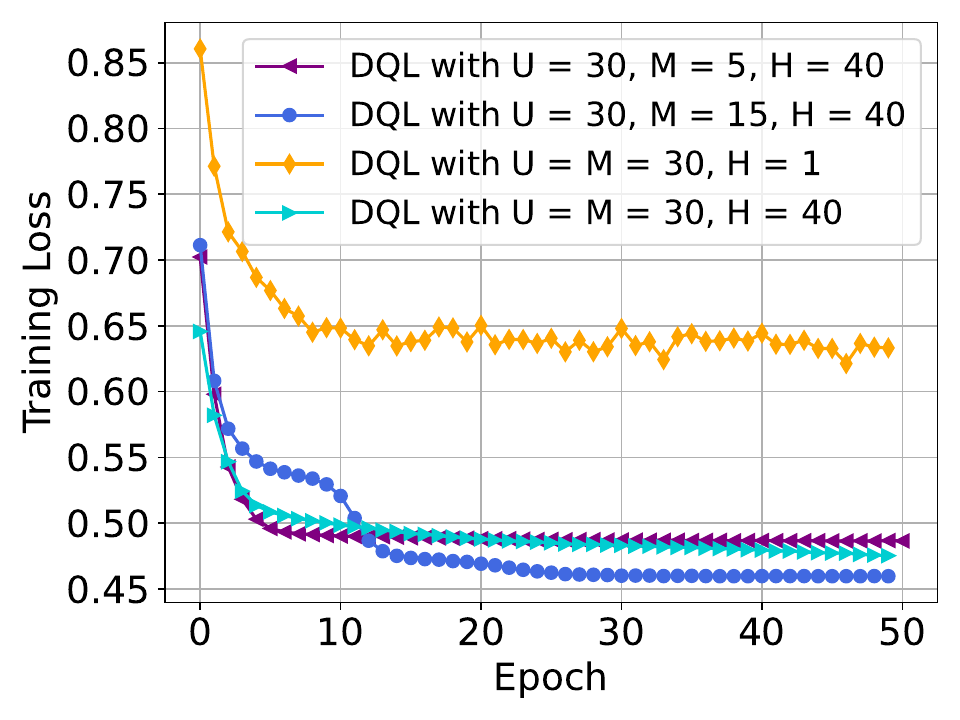}
        \caption{\footnotesize Training loss for higher measurement shots.}
        \label{fig1d}
    \end{subfigure}
    \vspace{-3pt}
    \caption{\footnotesize Training performance of DQL on MNIST (IID) under full and partial participation for different measurement shots, where $U$ denotes the total number of devices, $M$ the number of devices in the participating subset per round, and $H$ the number of measurement shots.}
    \label{fig1}
\end{figure}

\textcolor{black}{Fig.~\ref{fig1a} shows the impact of device participation on DQL training accuracy for the MNIST (IID) dataset with a single measurement shot. Here, IID refers to the setting where each client holds a representative subset of the global dataset, so local data and class distributions are statistically aligned.} 
\textcolor{black}{It is worth noting that the single-shot ($H=1$) experiments demonstrate the algorithm's robustness under maximum measurement uncertainty. While individual gradient estimates exhibit high variance, they remain unbiased estimators of the true gradient \cite{sweke2020stochastic}. Here, the optimization is based on temporal averaging across the training trajectory rather than spatial averaging over shots.} Furthermore, higher participation ($M = 30$) yields the lowest accuracy due to cumulative shot noise from aggregating all devices' updates, which amplifies errors in the global model. Partial participation alleviates this effect, using only 5 devices per epoch achieves the highest accuracy by limiting noise accumulation and stabilizing training. The 15-device case performs better than full participation, however, remains below the 5-device setup. These results reveal a clear trade-off between participation level and noise amplification, with fewer devices mitigating shot noise more effectively.

Fig.\ref{fig1c} presents the training loss for the same setup as Fig.\ref{fig1a}, mirroring the accuracy trends. Full participation ($M = 30$) incurs the highest loss due to accumulated shot noise, while reducing participation to 5 devices achieves the lowest loss through more stable convergence. The 15-device case performs better than full participation, however, worse than 5 devices. These results further confirm that limiting device participation can effectively mitigate shot noise, enhancing both stability and performance in QFL.

Fig.~\ref{fig1b} shows the effect of device participation and measurement shots on DQL accuracy for the MNIST (IID) dataset. With full participation ($M = 30$), accuracy is lowest at a single shot ($H = 1$) due to amplified shot noise, however, it improves notably at $H = 40$, indicating noise reduction with more shots. The best accuracy occurs with partial participation ($M = 15$, $H = 40$), balancing noise mitigation. While $M = 5$, $H = 40$ starts strong, it eventually lags behind $M = 15$, $H = 40$, showing that with reduced noise, contributions from more devices aid performance.

Fig.~\ref{fig1d} presents the corresponding loss trends. Full participation with $H = 1$ yields the highest loss, which decreases substantially at $H = 40$. The $M = 15$, $H = 40$ setup achieves the lowest loss, confirming that moderate participation with higher shots offers the best trade-off. In contrast, $M = 5$, $H = 40$ performs worse than $M = 15$, $H = 40$, reinforcing the benefit of aggregating more device updates once noise effects are minimized.

\textcolor{black}{Fig.~\ref{fig2} shows CIFAR-10 results under non-IID conditions. Here, non-IID refers to a skewed client data partition, where each client’s local label distribution is imbalanced relative to the global dataset (via a Dirichlet-based split with $\alpha_{\mathrm{Dir}}=0.5$). We generate this setting using a Dirichlet-based split with concentration parameter $\alpha_{\mathrm{Dir}}=0.5$, creating a moderate to high heterogeneity. This setting induces client drift, meaning that local gradients can deviate from the global objective. It typically increases training variance and slows convergence compared to the IID case.
Nonetheless, Fig.~\ref{fig2} reveals trends consistent with Fig.~\ref{fig1} in the sense that the best-performing participation level is not fixed. Rather, the ``optimal'' configuration shifts with the measurement-shot budget $H$. When $H$ is low, a smaller $M$ can be suitable because aggregating many highly noisy client updates accumulates shot noise. For higher $H$, the per-client update variance drops. Thus, a larger $M$ becomes preferable to accommodate more diverse non-IID data. This consistency across datasets and distribution settings highlights a core principle of DQL: \textit{optimal performance depends on balancing measurement shots with device participation}. The similarity in performance across datasets and distributions indicates that this trade-off is a fundamental property of DQL, robust to varying data types and settings.}

\begin{figure}[ht!]
    \centering
    \footnotesize
    \begin{subfigure}[t]{0.49\linewidth} 
        \centering
        \includegraphics[width=\linewidth]{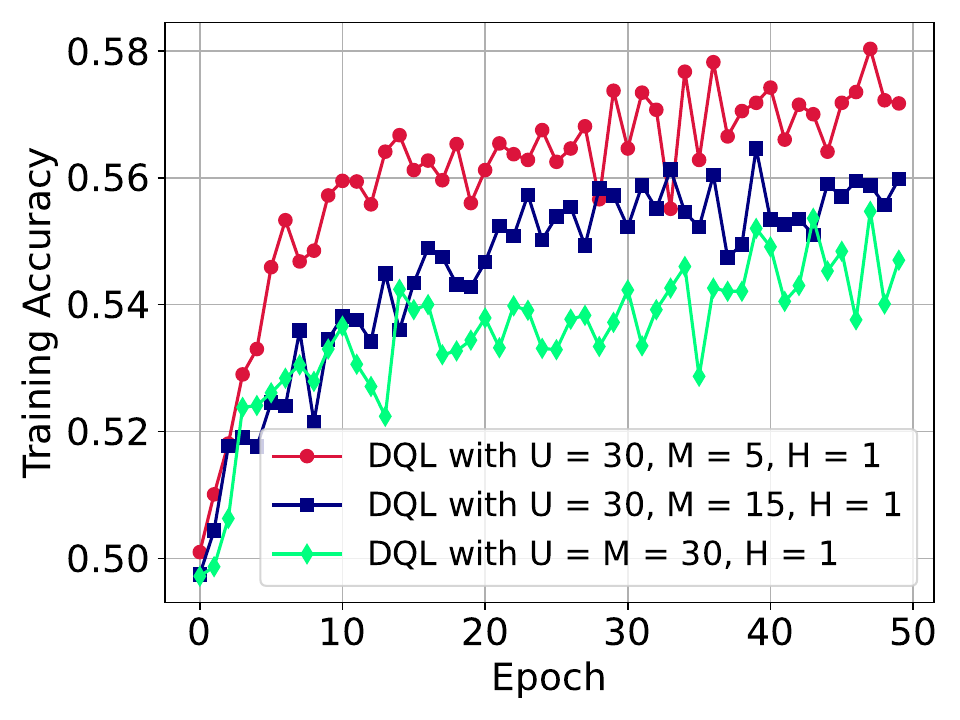}
        \caption{\footnotesize Training accuracy for a single measurement shot.}
        \label{fig2a}
    \end{subfigure}
    \hfill 
    \begin{subfigure}[t]{0.49\linewidth} 
        \centering
        \includegraphics[width=\linewidth]{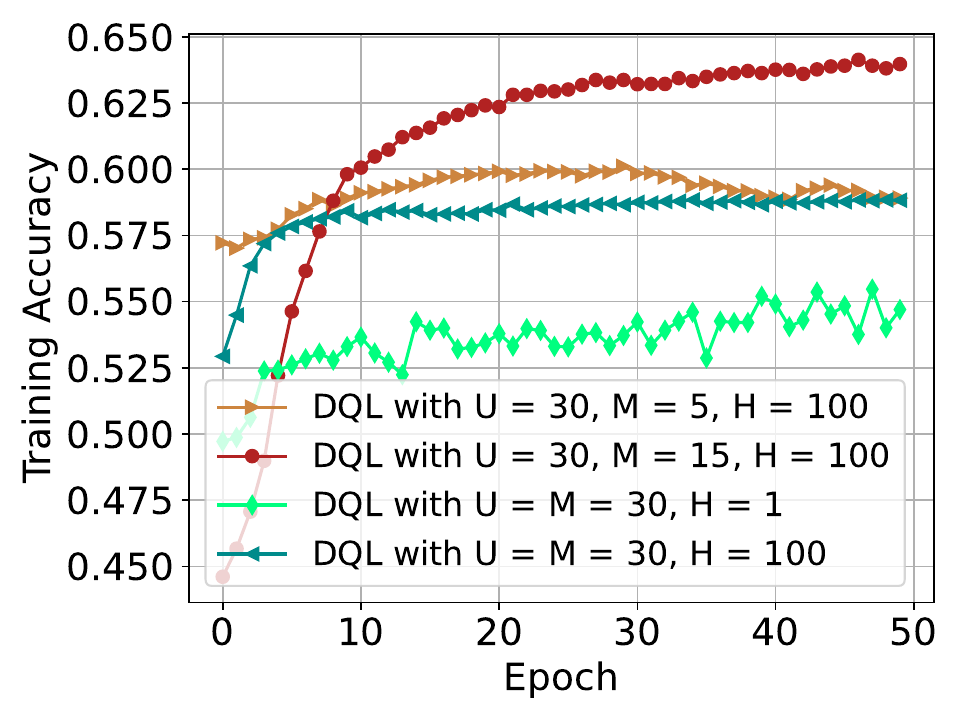}
        \caption{\footnotesize Training accuracy for single as well as higher measurement shots.}
        \label{fig2b}
    \end{subfigure}
    \vspace{5pt} 
    \begin{subfigure}[t]{0.49\linewidth} 
        \centering
        \includegraphics[width=\linewidth]{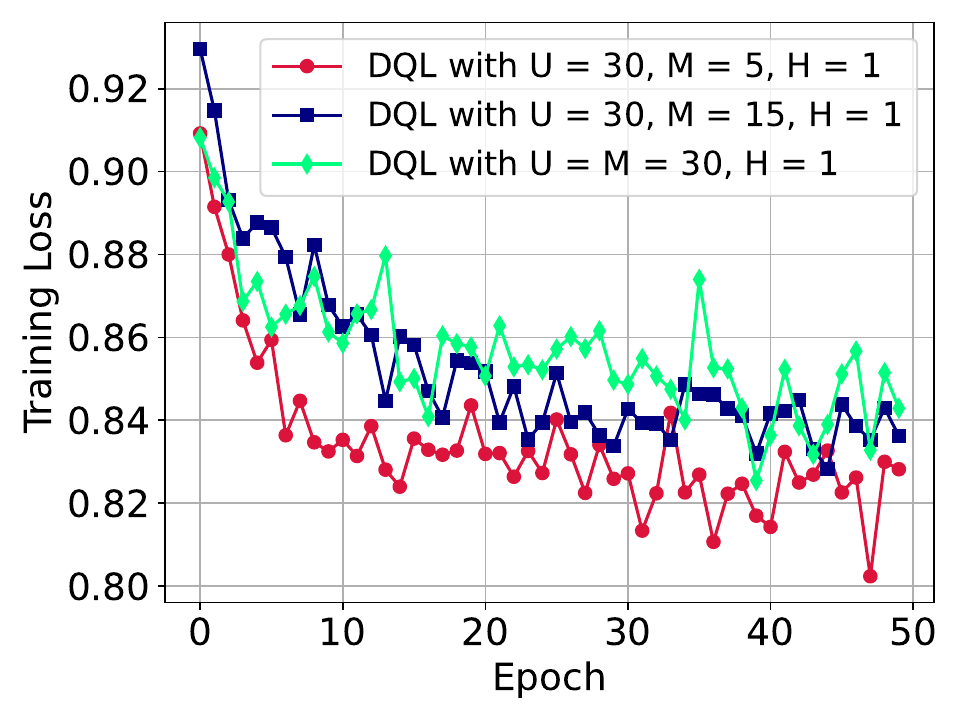}
        \caption{\footnotesize Training loss for single measurement shot.}
        \label{fig2c}
    \end{subfigure}
    \hfill 
    \begin{subfigure}[t]{0.49\linewidth} 
        \centering
        \includegraphics[width=\linewidth]{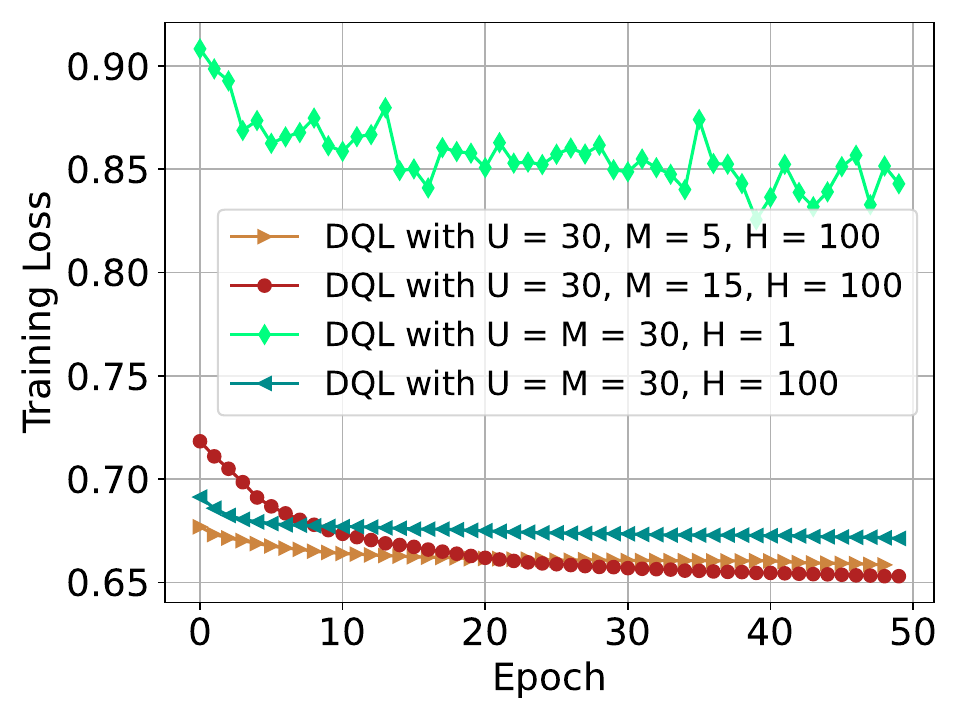}
        \caption{\footnotesize Training loss for single as well as higher measurement shots.}
        \label{fig2d}
    \end{subfigure}
    \vspace{-3pt}
    \caption{\footnotesize Training performance of DQL on CIFAR-10 dataset (non-IID) under full and partial client participation for different measurement shots, where $U$ denotes the total number of devices, $M$ the number of devices in the participating subset per round, and $H$ the number of measurement shots.}
    \label{fig2}
\end{figure}

\subsection{Security Analysis with Post-Quantum Cryptography}
\noindent\textbf{Threat Simulation and Detection:} 
The performance of this QNN-based threat detection model governs the transition between security levels. 
The comparative training performance of the proposed QNN and a conventional deep neural network (DNN) is illustrated in Fig. \ref{fig:threat detection training and test}. While both models ultimately converge to a similar high level of accuracy of over 90\%, the steeper training curve of QNN demonstrates its superior training efficiency, rapidly capturing the complex, non-linear relationships of the features.
The confusion matrix in Fig. \ref{fig:confusion matrix} further validates this robustness of QNN in testing data, revealing minimal misclassifications. 
\begin{figure}[htbp]
    \centering
    \begin{subfigure}[c]{0.54\linewidth} 
        \centering
        \includegraphics[width=\linewidth]{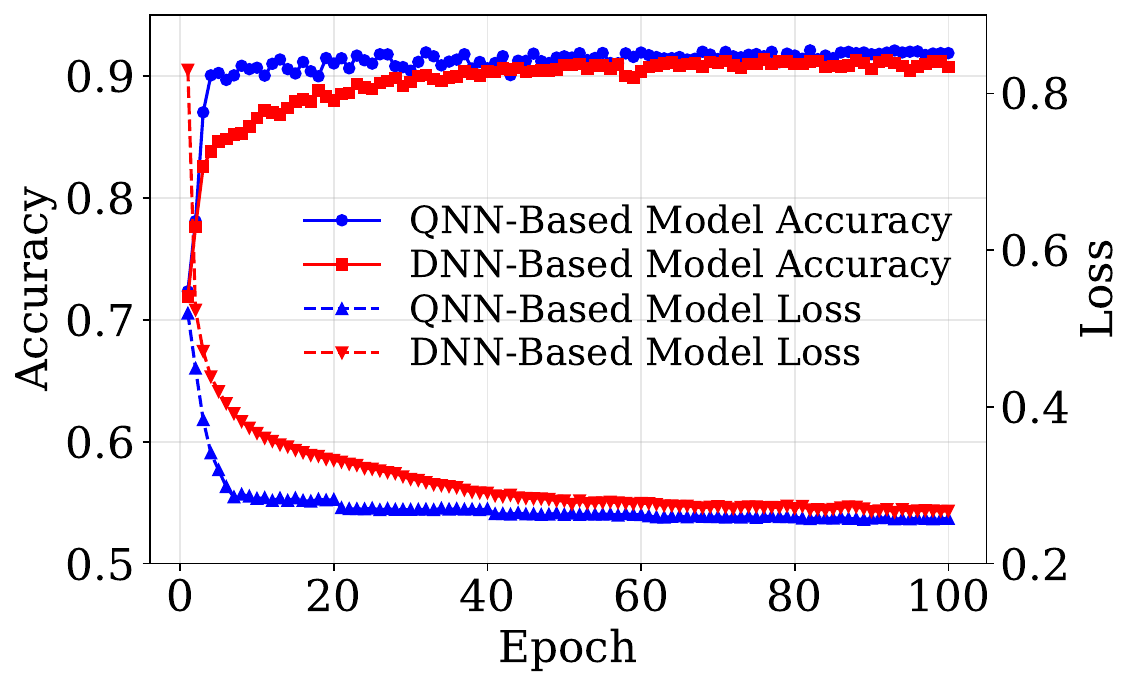}
        \vspace{-6mm}
        \caption{}
        \label{fig:threat detection training and test}
    \end{subfigure}
 \hfill 
    \begin{subfigure}[c]{0.43\linewidth} 
        \centering
        \includegraphics[width=\linewidth]{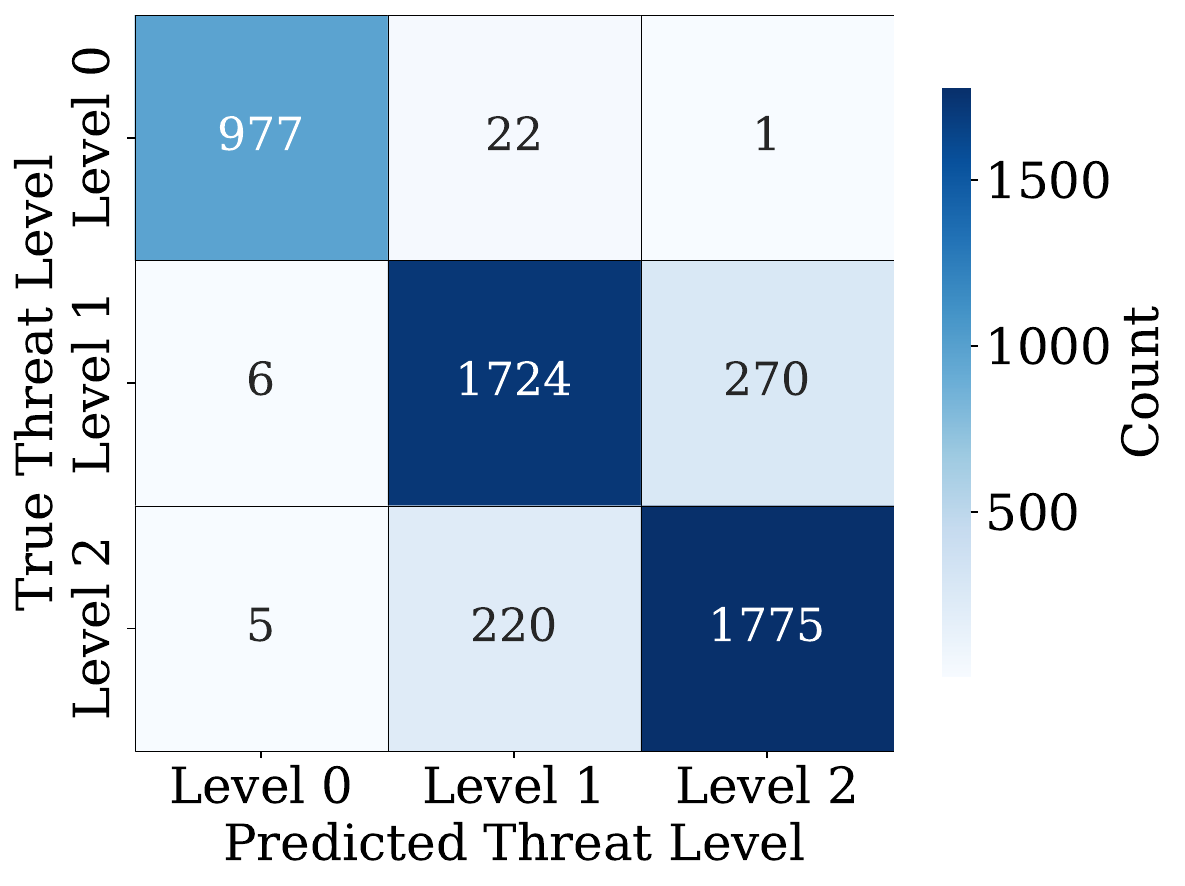}
        \caption{}
        \label{fig:confusion matrix}
    \end{subfigure}
    \caption{Performance evaluation of the threat detection model: (a) Training performance of QNN against baseline DNN. (b) Confusion matrix illustrating the classification results on a test set with QNN.}
    \label{fig:threat detection}
\end{figure}
Notably, these rare errors primarily occur between the two attack-related threat levels (Level 1 and Level 2), rather than misclassifying an attack as normal traffic. To enhance system stability and prevent false escalations, the final threat level is determined by monitoring the model's classification confidence over a window of trailing rounds, making the overall detection mechanism more robust.

\noindent\textbf{QS-DQL training under dynamic threats:}
\begin{figure}
    \centering
    \includegraphics[width=0.75\linewidth]{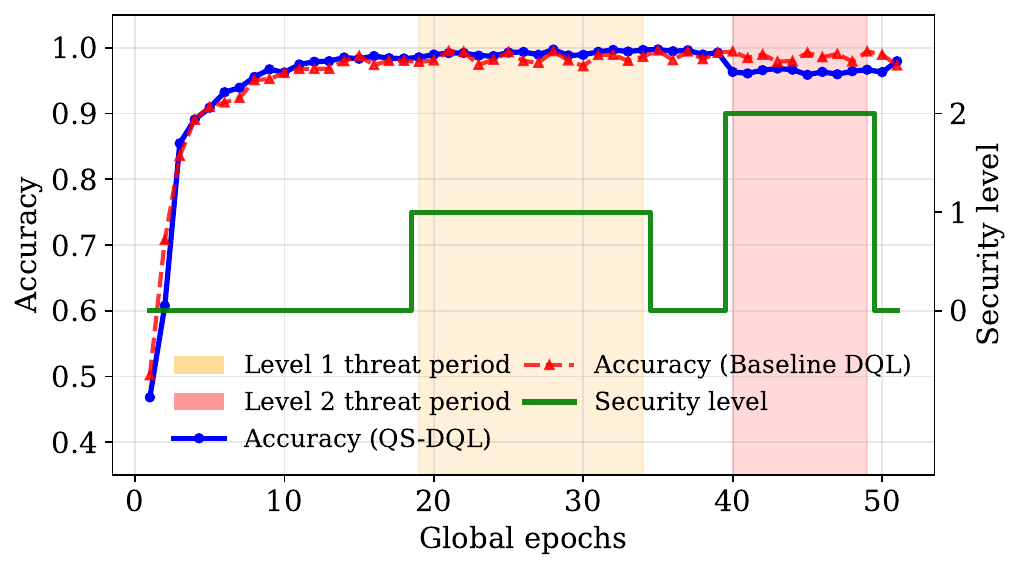}
    \caption{Performance of the adaptive QS-DQL framework under simulated threat scenarios. The plot shows the dynamic transition between security levels while the model's training accuracy continues to converge.}
    \label{fig:accuracy and security level switch}
\end{figure}

Fig. \ref{fig:accuracy and security level switch} demonstrates the training performance of the proposed QS-DQL framework. \textcolor{black}{In this experiment, we use OrganMNIST in a binary classification setting (two classes) to keep the task simple and isolate the effect of security-level transitions on training stability. Under this setting, the global model achieves a final accuracy of 98\% while seamlessly transitioning between security levels in response to simulated threats, which were generated by replicating network scenarios from the UNSW-NB15 dataset.} Throughout this training, the system also benefits from the fidelity-aware client participation technique. However, a slight depreciation in the model's convergence rate is observed during the Level 2 attack. This is an intentional design trade-off, where the framework overrides the optimal client selection and reduces participation to 50\% to stabilize the network under the high cryptographic load.

\noindent\textbf{Efficiency comparison of QS-DQL with baseline DQL:}
The resource consumption of the QS-DQL framework against static baselines with a fixed single security level is compared to quantify the advantage of the proposed method. Fig.~\ref{fig:cpu cycles over rounds} illustrates the average CPU cycles consumed by security operations on all client sides across all training rounds. For the correct instrumentation, the CPU frequency was set to a static 2.4 GHz and the encryption-related tasks were pinned to a single core. The static Level 2 baseline incurs the highest CPU cycles of (133 $\pm$25) Millions across all rounds, whereas Level 0 has the lowest (22$\pm$3 Millions) as expected, and Level 1 baseline consumes (63$\pm$15 Millions) CPU cycles. However, the adaptive QS-DQL model effectively enforces the Level 0 baseline during normal operations, only escalating its resource consumption during the simulated threat periods (rounds 20-30 and 40-50). This demonstrates the significant computational savings achieved by avoiding a constant high-security posture.
\begin{figure}
    \centering
    \includegraphics[width=0.75\linewidth]{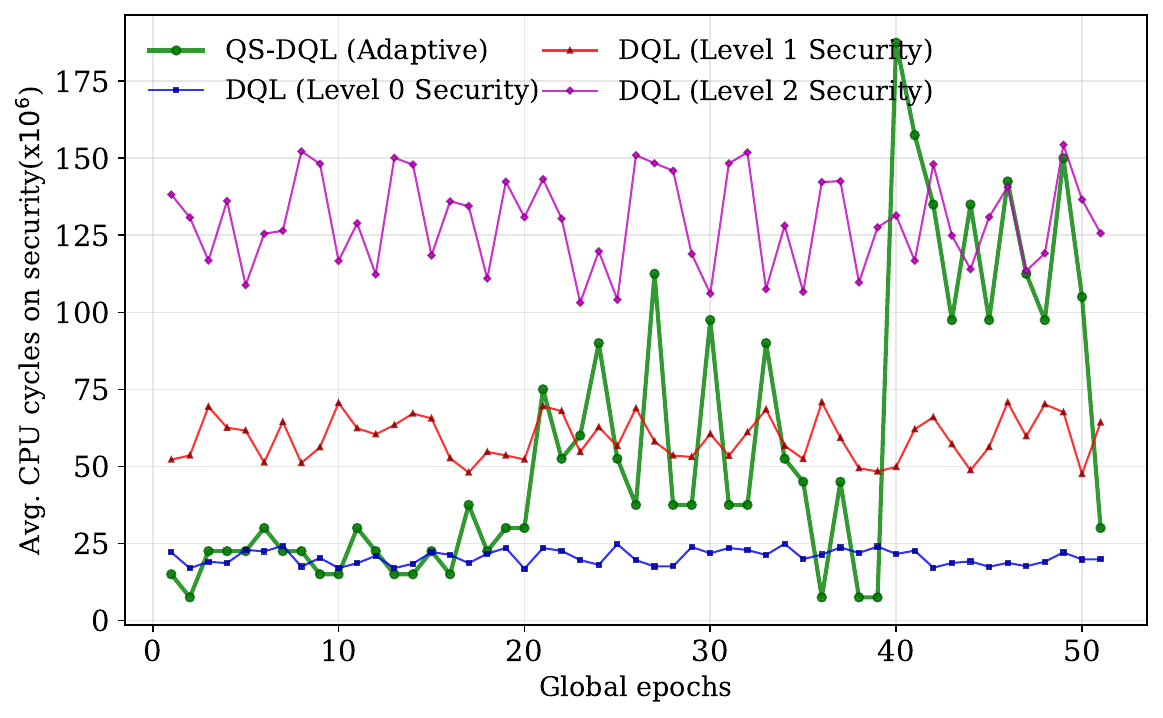}
    \caption{Comparison of average CPU cycles for security operations in proposed QS-DQL devices against static security level DQL baselines, evaluated across all training rounds.}
    \label{fig:cpu cycles over rounds}
\end{figure}

The average time taken for security operations in the entire client's plane is demonstrated in Fig.~\ref{fig:combined_security_time_clients}. Over 50 global rounds, a static Level 2 baseline consumes 1.95 times the total computational time on security operations compared to our adaptive framework. This bottleneck highlights the critical role of our adaptive framework in preserving overall system throughput.

\begin{figure}[htbp]
    \centering
    \begin{subfigure}[t]{0.495\linewidth}
        \centering
        \includegraphics[width=\linewidth]{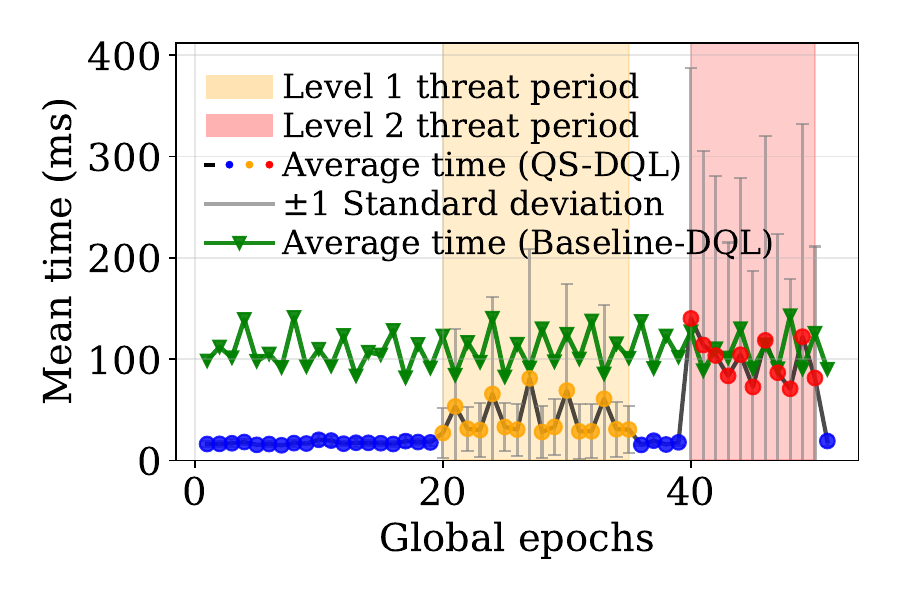}
        \caption{Average time consumed.}
        \label{fig:combined_security_time_clients}
    \end{subfigure}
    \hfill
    \begin{subfigure}[t]{0.465\linewidth}
        \centering
        \includegraphics[width=\linewidth]{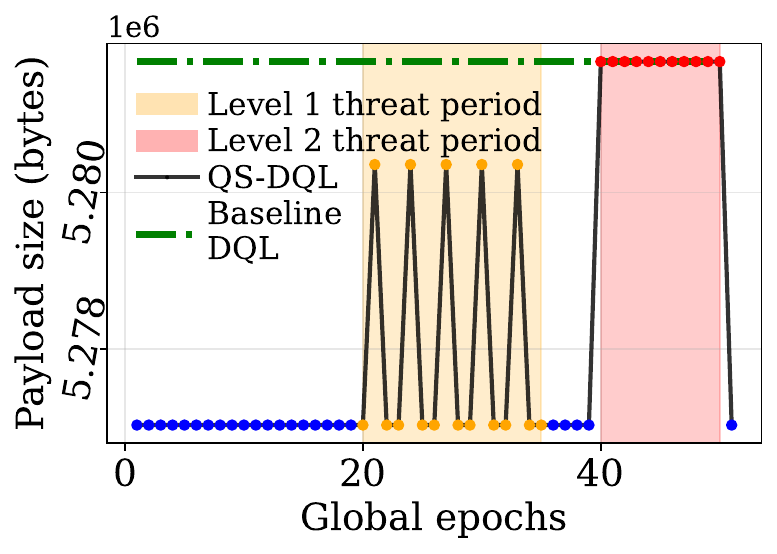}
        \caption{Average payload size.}
        \label{fig:average_payload_size}
    \end{subfigure}
    
    \caption{Comparison of (a) average time consumed for enforcing security mechanisms across all devices and (b) payload sizes from the client to server during model update over 50 global rounds in QS-DQL.}
    \label{fig:security_payload_comparison}
\end{figure}
\textcolor{black}{In our analysis, the focus is primarily on client-side computational overhead. We exclude server-side aggregation complexity, as this operation is linear in the number of model parameters $P$ and participating clients $M$ (i.e., $\mathcal{O}(MP)$) \cite{mcmahan2017communication}. This is typically considered computationally negligible compared to the iterative optimization times performed on client nodes \cite{luo2024communication}. }Furthermore, network latency is treated as out of scope due to deployment dependence, and we quantify this overhead using the per-round payload size. As shown in Fig.~\ref{fig:average_payload_size}, while the model size is dominant over security-related bytes, the PQC overhead is still measurable. On average, our adaptive framework reduces the data payload by 5 KB per transmission compared to a static Level 2 baseline. Although this overhead seems small in our experimental environment, its effect seems to be detrimental to a large-scale DQL network, urging the necessity of our proposed dynamic framework. 

\begin{table}[htbp]
\centering
\caption{CPU cycles for key generation (Kyber), signing (Dilithium), and encryption, along with Peak Memory Usage.}
\label{tab:security_cycles_memory}
\scriptsize
\setlength{\tabcolsep}{1pt}
\begin{tabular}{lccc}
\toprule
\shortstack[c]{Security Level\\ in \\QS-DQL}&
\shortstack[c]{Total Cycles\\($\times 10^9$)\\(Security Ops)} &
\shortstack[c]{Avg. Cycles/Round\\($\times 10^6$)\\(Security Ops)} &
\shortstack[c]{Peak Memory\\ Usage (KB)\\(Overall Training)} \\
\midrule
Level 0 (Normal) & 2.01 & 32.4  & 356{,}470 \\
Level 1 (Medium) & 3.78 & 74.1  & 356{,}520 \\
Level 2 (High)   & 5.67 & 171.8 & 357{,}820 \\
\bottomrule
\end{tabular}
\end{table}

To profile the stress over a single client, Table \ref{tab:security_cycles_memory} details the total and average CPU cycles consumed exclusively by PQC operations during dynamic escalation, alongside the overall peak memory usage during training. The results clearly show a direct correlation between the security level and resource cost. Compared to Level 0, the computational cycles increase by over 2.2 times at Level 1 and more than 5.3 times at Level 2, while peak memory usage also shows a slight increase as the framework transitions. This overhead is due to the stronger PQC primitives, as the higher-level variants of Kyber and Dilithium operate on larger matrices and polynomials.



\begin{figure}[htbp]
    \centering
    \begin{subfigure}[b]{0.485\linewidth}
        \centering
        \includegraphics[width=\linewidth]{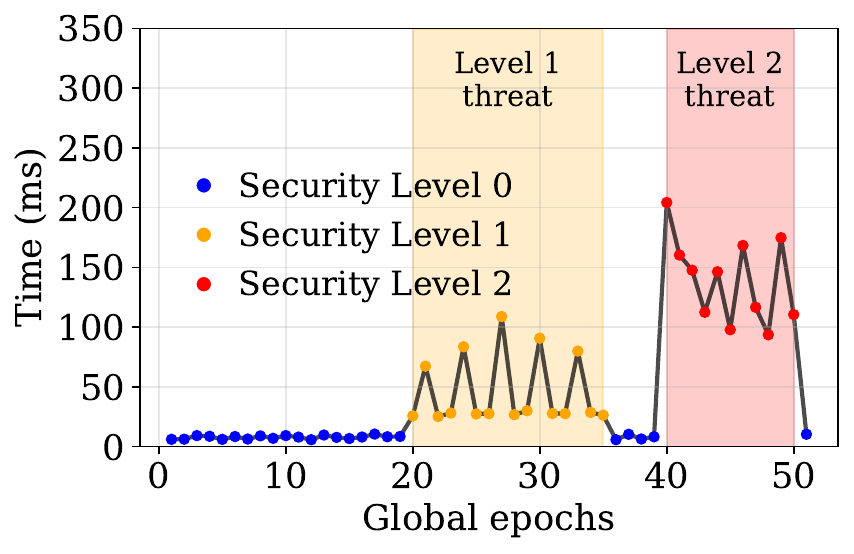}
        \caption{Time consumed by security operations in a single client.}
        \label{fig:pure_security_time}
    \end{subfigure}
    \hfill
    \begin{subfigure}[b]{0.475\linewidth}
        \centering
        \includegraphics[width=\linewidth]{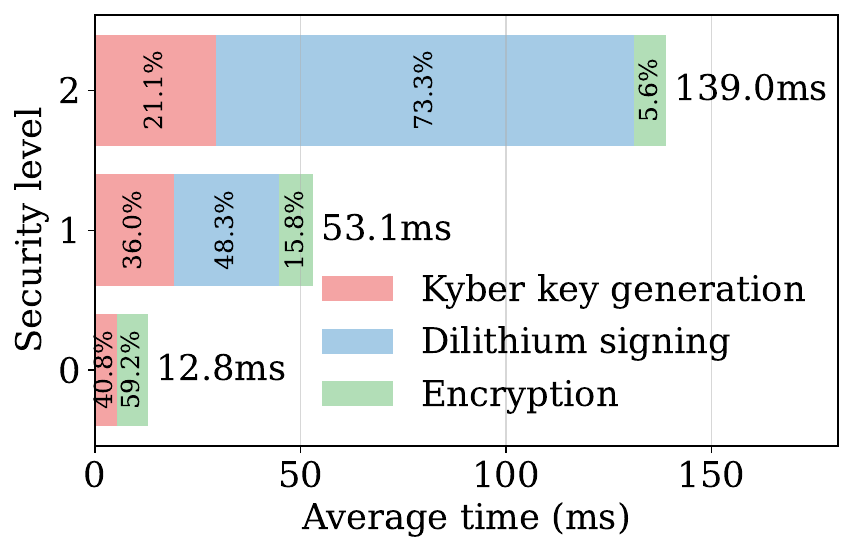}
        \caption{Operation-wise breakdown of security levels. }
        \label{fig:security_breakdown}
    \end{subfigure}
    \caption{Comparison of (a) time consumed by security mechanism and (b) breakdown of security operations in a single client.}
    \label{fig:single_client_security_comparison}
\end{figure}

Fig. \ref{fig:pure_security_time} isolates the time spent exclusively on cryptographic operations for a single client within the QS-DQL framework. The breakdown of this overhead is detailed in Fig. \ref{fig:security_breakdown}. At Level 2, Dilithium signing emerges as the most computationally expensive task, dominating the total security time. This occurs even though both signing and re-keying (Kyber) are performed in every round, highlighting the inherent cost of the signature algorithm. In contrast, at Level 1, where both PQC operations are invoked periodically, their time contributions are comparable. Across all security levels, the symmetric encryption consistently represents the smallest fraction of the computational overhead.

Collectively, these results validate that our adaptive framework effectively manages security overheads to maintain system efficiency. However, this framework has limitations that future research should address. The reliance on a central server creates a potential single point of failure, urging the use of decentralized architectures. Furthermore, the threat model is currently focused on communication security, as one mechanism cannot protect against all kinds of threats; future work should expand this to include adversarial attacks on the quantum models.

\section{Conclusion}

This paper introduced a two-stage framework for DQL that jointly targets convergence optimization and adaptive security. In Stage~1, we analyzed convergence under partial device participation with non-convex objectives and heterogeneous data, and derived an upper bound that exposes a fundamental trade-off among convergence behavior, the measurement-shot budget, and the participating device subset size. In Stage~2, we proposed a multi-layered PQC architecture with a QNN-driven adaptive mechanism that monitors operating conditions, assesses potential threats, and dynamically selects across three NIST-aligned security levels to control cryptographic overhead. Numerical simulations validate the theoretical trends and demonstrate that the proposed design can mitigate security costs while maintaining stable learning across diverse data and network settings. \textcolor{black}{As future work, we will extend this framework toward real deployment by decentralizing aggregation and threat detection via consensus-based peer-to-peer updates and hierarchical edge-assisted aggregation. The fidelity-aware participation mechanism will be retained through local neighbor-selection strategies. We will expand the threat model to stronger adversaries (e.g., backdoor and poisoning attacks) and integrate secure aggregation with robust aggregation rules to withstand malicious or unreliable clients. For a more thorough evaluation of convergence behavior, future studies will incorporate calibrated hardware-realistic noise channels from real quantum processors. The proposed QNN-based DQL pipeline will then be deployed on real quantum hardware under practical shot budgets and time-varying device conditions.}

\bibliographystyle{IEEEtran}
\bibliography{Main_Bibliography}

\appendix

\noindent
\textbf{\textit{Detailed Derivations}}
We now present proofs of lemmas and theorem used throughout our paper.

\subsection{Proof of Lemma \ref{lemma1}} \label{prooflemma1}
As mentioned before, let $\mathcal{U} = \{1, 2, \dots, U\}$ denote the set of devices, and let $\Tilde{g}_{k}^{t} = \frac{1}{U} \sum_{u \in \mathcal{U}} \Tilde{g}_{u,k}^{t}$ represent the average of their local stochastic gradients at local iteration $t$ at global round $k$. We have
\begin{align}
    &- \mathbb{E}_{\{\xi_{1,k}^{t}, \dots, \xi_{u,k}^{t} | \boldsymbol{w}_{1,k}^{t}, \dots, \boldsymbol{w}_{U,k}^{t}\}} \mathbb{E}_{\{1, 2, \dots\, U \}\in \mathcal{U}} \bigg[\langle \nabla f(\bar{\boldsymbol{w}}_{k}^{t}), \Tilde{g}_{k}^{t} \rangle\bigg] \nonumber \\
    &\stackrel{\textcircled{\footnotesize 1}}{=} - \mathbb{E}_{\{1, 2, \dots\, U \}\in \mathcal{U}} \nonumber \\
    & \hspace{2.5em} \mathbb{E}_{\{\xi_{1,k}^{t}, \dots, \xi_{u,k}^{t} | \boldsymbol{w}_{1,k}^{t}, \dots, \boldsymbol{w}_{U,k}^{t}\}}  \bigg[\langle \nabla f(\bar{\boldsymbol{w}}_{k}^{t}), \frac{1}{U}\sum_{u \in \mathcal{U}}\Tilde{g}_{u,k}^{t} \rangle\bigg] \nonumber \\
    &\stackrel{\textcircled{\footnotesize 2}}{\leq} \frac{1}{2} \bigg[ 
    -||\nabla f(\bar{\boldsymbol{w}}_{k}^{t})||_{2}^{2} - || \sum_{u=0}^{U} \nabla f_{u}(\boldsymbol{w}_{u,k}^{t})||_{2}^{2} \nonumber \\
    & \hspace{12em} + \sum_{u=0}^{U} L^{2} || \bar{\boldsymbol{w}}_{k}^{t} - \boldsymbol{w}_{u,k}^{t}||_{2}^{2}\bigg],
\end{align} 

where \textcircled{\footnotesize 1} is due to the fact that random variables $\xi_{u,k}^{t}$ and $\mathcal{N}$ are independent, and \textcircled{\footnotesize 2} follows from Assumption 0.1.

\subsection{Proof of Lemma \ref{lemma2}} \label{prooflemma2}
Let $\{x_i\}_{i=1}^U$ denote any fixed deterministic sequence. We sample a multiset $\mathcal{M}$ (with size $M$) by the procedure where for each sampling time, we sample $x_{u}$ with probability $q_{u}$ for each time. Let $\mathcal{M} = \{i_{1}, i_{2}, \dots, i_{M}\} \subset \mathcal{U}$ (some $i_{u}$’s may have the same value). Then
\begin{align}
\mathbb{E}_{\mathcal{M}} \left[ \sum_{u \in \mathcal{M}} x_{u} \right] &= \mathbb{E}_{\mathcal{M}} \left[ \sum_{u=1}^{M} x_{i_{u}} \right] = M \sum_{u=1}^{U} q_{u} x_{u}.
\end{align}
It is easy to prove this lemma when equipped with this observation.

\subsection{Proof of Lemma \ref{lemma3}} \label{prooflemma3}
\noindent
Let $\mathcal{M} = \{i_{1}, i_{2}, \dots, i_{M}\}$ denote the multiset of chosen indexes. We define $\bar{\boldsymbol{w}}_{k+1} = \frac{1}{M} \sum_{l=1}^{M} v_{i_{l}}^{k+1}$. Taking expectation over $\mathcal{M}$, we have $\mathbb{E}_{\mathcal{M}} || \bar{\boldsymbol{w}}_{k+1} -\bar{\mathbf{v}}_{k+1} ||^{2} = \frac{1}{M} \sum_{u=1}^{U} q_{u} || \mathbf{v}_{u,k+1} - \bar{\mathbf{v}}_{t+1} ||^2,$
Since $k+1$ is a global round, we know that $k_{0} = k$ is also a global communication round, which implies $\{\boldsymbol{w}_{u,k_{0}}\}_{u=1}^U$ is identical. Then $\sum_{u=1}^{U} q_{u} || \mathbf{v}_{u,k+1} - \bar{\mathbf{v}}_{t+1} ||^2 \leq \sum_{u=1}^{U} q_{u} || (\mathbf{v}_{u,k+1} - \bar{\boldsymbol{w}}_{k_{0}})||^2$, where $\sum_{u=1}^U q_{u} \left( \mathbf{v}_{u,k+1} - \bar{\boldsymbol{w}}_{k_{0}} \right) = \mathbf{v}_{u,k+1} - \bar{\boldsymbol{w}}_{k_{0}}$ and $\mathbb{E} || x - \mathbb{E}x ||^{2} \leq \mathbb{E}||x||^{2}$. Similarly, we have
\begin{align}
&\mathbb{E}_{\mathcal{M}} \mathbb{E} || \bar{\boldsymbol{w}}_{k+1} -\bar{\mathbf{v}}_{k+1} ||^{2} \nonumber \\
&\leq \frac{1}{M} \sum_{u=1}^{M} q_{u} E \sum_{i=k_{0}}^{k} \mathbb{E} || \eta_{i} \nabla F_{u} (\boldsymbol{w}_{u,i}, \xi_{i, u})||^2 \leq \frac{4}{M} \eta_{k}^{2} T^{2} G^{2},
\end{align}
where in the last inequality we use the fact that $\eta_{k}$ is non-increasing and $\eta_{k_{0}} \leq 2\eta_{k}$. It concludes the proof.

\subsection{Proof of Lemma \ref{lemma4}} \label{prooflemma4}
\noindent
Since $\mathbb{E} \Tilde{g}_{k}^{t} = g_{k}^{t}$ and $\Tilde{g}_{u,k}^{t} = g_{u,k}^{t} + \xi_{u,k}^{t}$, where $\Tilde{g}_{u,k}^{t}$ is the stochastic gradient estimate, $g_{u,k}^{t}$ the true gradient, and $\xi_{u,k}^{t}$ the noise term arising from the limited number of measurements in practical PQC execution. Here, $\xi_{u,k}^{t}$ satisfies $\mathbb{E}[\xi_{u,k}^{t}] = 0$ and $\text{var}(\xi_{u,k}^{t}) = \mathbb{E}[\Vert \Tilde{g}_{u,k}^{t} - g_{u,k}^{t}\Vert^2 ]$. The variance of the stochastic gradient estimate is
\begin{align}
    &\mathbb{E}\bigg[||\Tilde{g}_{k}^{t} - g_{k}^{t}||^{2}\bigg] =\frac{C_{1}}{U^2}\sum_{u=1}^{U}\Big[\big\|g_{u,k}^{t}\big\|^2 + \frac{\sigma^2}{UB} + \frac{\text{var}(\xi_{u,k}^{t})}{U}\Big] \nonumber \\
\end{align}
\begin{align}
    &\mathbb{E}\bigg[\big\| \Tilde{g}_{k}^{t}\big\|^2\bigg] 
    = \mathbb{E}\bigg[\big\| \Tilde{g}_{k}^{t} - \mathbb{E}[\Tilde{g}_{k}^{t}]\big\|^2\bigg] + \big\|\mathbb{E}[\Tilde{g}_{k}^{t}]\big\|^2 \nonumber \\
    &\leq \bigg(\frac{C_{1}+U}{U^2}\bigg)\sum_{u=1}^{U}\big\|g_{k}^{t}\big\|^2 + \frac{\sigma^2}{U|\xi_{u}^{(\ell),t}|} + \frac{\text{var}(\gamma^{(\ell),t})}{U}.
\end{align}
Using the upper bound on the weighted gradient diversity, $\sum_{u=1}^{U}\big\|g_{u,k}^{t}\big\|^2 \leq \lambda U \big\|g_{u,k}^{t}\big\|^2$, we get
\begin{align}
    &\mathbb{E}\bigg[||\Tilde{g}_{k}^{t}||^2\bigg] \leq \bigg(\frac{C_{1}}{U}+1\bigg) \bigg[\sum_{u=1}^{U}||\nabla f_{u}(\boldsymbol{w}_{u,k}^{t})||^{2}\bigg] + \frac{\sigma^{2}}{UB} \nonumber \\ 
    &\leq \lambda \bigg(\frac{C_{1}}{U}+1\bigg) \bigg[\sum_{u=1}^{U}||\nabla f_{u}(\boldsymbol{w}_{u,k}^{t})||^{2}\bigg] + \frac{\sigma^{2}}{UB} + \frac{\text{var}(\xi_{u,k}^{t})}{M},
\end{align}
where $B$ is the mini-batch size, completing the proof. 

\subsection{Proof of Lemma \ref{lemma5}} \label{prooflemma5}
\noindent \textcolor{black}{A critical quantum-specific derivation occurs in this proof. Unlike classical variance, which scales with data heterogeneity, the gradient estimate variance here is derived utilizing the {parameter-shift rule ($\pm \pi/2$)} and the {Cauchy-Schwarz inequality} applied to the observable's eigendecomposition. This quantifies the statistical uncertainty present in collapsing the wavefunction into discrete eigenvalues during measurement.}

We fix global and local iteration indices $k$ and $t$, and omit them for brevity temporarily. We define $X_{u,d,\pm} = {\langle \hat{Z} \rangle}_{\lvert \Psi_{u}(\boldsymbol{w}_{u} \pm \frac{\pi}{2} e_{d}) \rangle} - {\langle Z \rangle}_{\lvert \Psi_{u}(\boldsymbol{w}_{u} \pm \frac{\pi}{2} e_{d}) \rangle}$, where $\lvert \Psi_{u,d,\pm} \rangle = {\lvert \Psi_{u}(\boldsymbol{w}_{u} \pm \frac{\pi}{2} e_{d}) \rangle}$ is the quantum state with the $d^{\text{th}}$ parameter shifted by $\pm \frac{\pi}{2}$. The variance of the gradient estimate is written as $\text{var}(\xi_{u}) = \mathbb{E} [ \sum_{d=1}^{D} ( \frac{1}{2} ( {\langle \hat{Z} \rangle}_{\lvert \Psi_{u,d,+} \rangle} - {\langle \hat{Z} \rangle}_{\lvert \Psi_{u,d,-} \rangle}) \hspace{1em} - \frac{1}{2} ( {\langle Z \rangle}_{\lvert \Psi_{u,d,+} \rangle} - {\langle Z \rangle}_{\lvert \Psi_{u,d,-} \rangle}) )^{2} ] = \sum_{d=1}^{D} \frac{1}{4}  ( \mathbb{E}[X_{u,d,+}^{2}] - \mathbb{E}[X_{u,d,-}^{2}] )$,
where the expectation is over $H$ measurements of the quantum states ${\lvert \Psi_{u}(\boldsymbol{w}_{u} + \frac{\pi}{2} e_{d}) \rangle}$ and ${\lvert \Psi_{u}(\boldsymbol{w}_{u} - \frac{\pi}{2} e_{d}) \rangle}$ for $d = 1, 2, \dots, D$. Thus, $X_{u,d,+}$ and $X_{u,d,-}$ are independent, giving the inequality. We note that $\mathbb{E}[X_{u,d,+}^{2}]$ = $\text{var}({\langle \hat{Z} \rangle}_{\lvert \Psi_{u,d,+} \rangle})$. Let $Y$ index the measurement outcome of observable $Z$, so that $Z = h_{Y}$, and define $W_{y} = \mathbb{I}\{Y = y\}$ determining whether $Y = y (W_{y} = 1)$ or not $(W_{y} = 0)$. Therefore, it follows from the definition ${\langle \hat{Z} \rangle}_{\lvert \Psi_{u,d,+} \rangle}$ that $\mathbb{E}[X_{u,d,+}^{2}] = \frac{1}{H} \text{var}( \sum_{y=1}^{N_{z}} h_{y} W_{y} ) = \frac{1}{H} \mathbb{E} [ ( \sum_{y=1}^{N_{z}} h_{y} (W_{y} - p(y|\boldsymbol{w}_{u} + e_{d} \frac{\pi}{2})) )^{2} ] \leq \frac{N_{z}}{N_{y}} (\sum_{y=1}^{N_{z}} h_{y}^{2}) v = \frac{N_{z} \operatorname{Tr}(Z^{2})}{H} v$, where \textcircled{\footnotesize 1} is from Cauchy-Schwarz inequality, \textcircled{\footnotesize 2} occurs due to the fact that the variance of $W_{y}$ is computed as $\text{var}(W_{y}) = \mathbb{E}[W_{y}^{2}] - (\mathbb{E}[W_{y}])^{2} = v(p(y|\boldsymbol{w}_{u} + e_{d} \frac{\pi}{2}))$,
where $v(x) = x (1-x)$ for $x \in (0,1)$. The last term arises from the definition of $v$. Similarly, it can be shown that the following inequality holds $\mathbb{E}[X_{u,d,-}^{2}] \leq \frac{N_{z} \operatorname{Tr}(Z^{2})}{H} v$. We can write while bringing the omitted indices back $\text{var}(\xi_{u,k}^{t}) \leq \frac{\nu N_{z} D Tr(Z^{2})}{2H}$. For $U$ number of DQL clients, we get $\text{var}(\xi_{k}^{t}) \leq \frac{1}{U} \sum_{u \in \mathcal{U}} \frac{\nu N_{z} D Tr(Z^{2})}{2H}$, concluding the proof.

\subsection{Proof of Lemma \ref{lemma6}} \label{prooflemma6}
\noindent
If $k = i_{c}$ is the most recent global communication round, $\bar{\boldsymbol{w}}^{i_{c}+1} = \frac{1}{U} \sum_{u \in \mathcal{U}} \boldsymbol{w}_{u}^{i_{c + 1}}$. The local solution at device $u$ at any particular iteration $i > i_{c}$, where $i$ is assumed to represent the most recent iteration, encompassing all global and local iterations up to the current point, can be written as: $\boldsymbol{w}_{u,k}^{t} = \boldsymbol{w}_{u}^{i} = \boldsymbol{w}_{u}^{i-1} - \eta_{i_{c}} \Tilde{g}_{u}^{i-1} = \bar{\boldsymbol{w}}^{i_{c}+1} - \sum_{z=i_{c}+1}^{i-1} \eta_{i_{c}} \Tilde{g}_{u}^{z}$. Now, we compute the average virtual model at iteration $i$ as follows: $\bar{\boldsymbol{w}}^{i} = \bar{\boldsymbol{w}}^{i_{c}+1} - \frac{1}{U} \sum_{u \in \mathcal{U}} \sum_{z=i_{c}+1}^{i-1} \eta_{i_{c}} \Tilde{g}_{u}^{z}$. Firstly, without loss of generality, suppose $i = s_{t} T + r$, with $s_{t}$ and $r$ denoting the indices of global communication round and local updates, respectively. Next, we consider that for $i_{c}+1 < i \leq i_{c} + T$, $\mathbb{E}_{i}||\bar{\boldsymbol{w}}^{i} - \boldsymbol{w}_{u}^{i}||$ does not depend on time $i \leq i_{c}$ for $1 \leq u \leq U$. Therefore, for all iterations $1 \leq i \leq I$, where $I = KT$, we can write, $\frac{1}{KT} \sum_{k=1}^{K} \sum_{t=1}^{T} \sum_{u=1}^{U} \mathbb{E}||\bar{\boldsymbol{w}}_{k}^{t} - \boldsymbol{w}_{u,k}^{t}||^{2} = \frac{1}{I} \sum_{i=1}^{I} \sum_{u=1}^{U} \mathbb{E}||\bar{\boldsymbol{w}}^{i} - \boldsymbol{w}_{u}^{i}||^{2} \hspace{4em} = \frac{1}{I} \sum_{s_{t}=1}^{\frac{I}{T} - 1} \sum_{r=1}^{T} \sum_{u=1}^{U} \mathbb{E}||\bar{\boldsymbol{w}}^{s_{t} E + r} - \boldsymbol{w}_{u}^{s_{t} E + r}||^{2}$.
We bound the term $\mathbb{E}||\bar{\boldsymbol{w}}^{i} - \boldsymbol{w}_{l}^{i}||^{2}$ for $i_{c}+1 \leq i = s_{t} T + r \leq i_{c} + T $ in threes steps: (1) We first relate this quantity to the variance between stochastic gradient and full gradient, (2) We use Assumption \ref{Assumption1}, (3) We use Assumption \ref{Assumption3} to bound the final terms. Here, $l$ is associated with an individual device, while $u$ is used for summing over devices.

\vspace{-8pt}
\begin{align}
    &\mathbb{E}||\bar{\boldsymbol{w}}^{s_{t} E + r} - \boldsymbol{w}_{l}^{s_{t} E + r}||^{2} \nonumber \\
    &\stackrel{\textcircled{\footnotesize 1}}{=} \mathbb{E}|| \sum_{z=1}^{r} \eta_{i_{c}} \Tilde{g}_{l}^{s_{t}+z} - \frac{1}{U} \sum_{u \in \mathcal{U}} \sum_{z=1}^{r} \eta_{i_{c}} \Tilde{g}_{u}^{s_{t}+z}||^{2} \nonumber \\
    &\hspace{4em} - \mathbb{E} ||\sum_{z=1}^{r} \eta_{i_{c}} \Tilde{g}_{l}^{s_{t}+z}||^{2} + \mathbb{E} ||\frac{1}{U} \sum_{u \in \mathcal{U}} \sum_{z=1}^{r} \eta_{i_{c}} \Tilde{g}_{u}^{s_{t}+z} \nonumber \\
    &- \mathbb{E}\bigg[\frac{1}{U} \sum_{u \in \mathcal{U}} \sum_{z=1}^{r} \eta_{i_{c}} \Tilde{g}_{u}^{s_{t}+z} \bigg] ||^{2} \bigg] + ||\mathbb{E}\bigg[\frac{1}{U} \sum_{u \in \mathcal{U}} \sum_{z=1}^{r} \eta_{i_{c}} \Tilde{g}_{u}^{s_{t}+z}\bigg]||^{2} \nonumber \\
    &\hspace{10em} + ||\frac{1}{U} \sum_{u \in \mathcal{U}} \sum_{z=1}^{r} \eta_{i_{c}} g_{u}^{s_{t}T+z}||^{2}\bigg),
\end{align}
where \textcircled{\footnotesize 1} holds because $i = s_{t}T+r \leq i_{c} + T$, \textcircled{\footnotesize 2} comes from Assumption 0.1.
\begin{align}
    &= 2 \bigg( \bigg[ \sum_{z=1}^{r} \eta_{i_{c}}^{2} \mathbb{E} ||\Tilde{g}_{l}^{s_{t}T+z} - g_{l}^{s_{t}T+z}||^{2} + r \sum_{z=1}^{r}\eta_{i_{c}}^{2} \mathbb{E} ||g_{l}^{s_{t}T+z}||^{2}\bigg] \nonumber \\
    & + \frac{1}{U^{2}} \sum_{u \in \mathcal{U}} \sum_{z=1}^{r} \eta_{i_{c}}^{2} \mathbb{E} ||\Tilde{g}_{u}^{s_{t}T+z} - g_{u}^{s_{t}T+z}||^{2} \nonumber \\
    & + \frac{r}{N^{2}}  \sum_{u \in \mathcal{U}} \sum_{z=1}^{r} \eta_{i_{c}}^{2} \mathbb{E} ||g_{u}^{s_{t}T+z} ||^{2} \bigg), \label{eqn55}
\end{align}
Our next step is to bound the terms in \eqref{eqn55} using Assumption 3 as follows:
$\mathbb{E} ||\bar{\boldsymbol{w}}_{k}^{t} - \boldsymbol{w}_{l,k}^{t}||^{2} \leq = 2 ( [ \sum_{z=1}^{r} \eta_{i_{c}}^{2} C_{1} ||g_{l}^{s_{t}T+z}||^{2} + \sum_{z=1}^{r} \eta_{i_{c}}^{2} \frac{\sigma^{2}}{B} + r \sum_{z=1}^{r} \eta_{i_{c}}^{2} || g_{l}^{s_{t}T+z}||^{2} ] + \frac{1}{U^{2}} \sum_{u \in \mathcal{U}} \sum_{z=1}^{r} \eta_{i_{c}}^{2} C_{1} ||g_{u}^{s_{t}T+z}||^{2} + \sum_{z=1}^{r} \eta_{i_{c}}^{2} \frac{\sigma^{2}}{UB} + \frac{r}{U^{2}} \sum_{u \in \mathcal{U}} \sum_{z=1}^{r} \eta_{i_{c}}^{2} || g_{u}^{s_{t}T+z}||^{2} )$.
Now we determine the upper bound for $\sum_{r=1}^{T} \sum_{u=1}^{U} [\mathbb{E} ||\bar{\boldsymbol{w}}_{k}^{t} - \boldsymbol{w}_{u,k}^{t}||]$ using as follows: $\sum_{r=1}^{T} \sum_{u=1}^{U} [\mathbb{E} ||\bar{\boldsymbol{w}}^{s_{t}T+z} - \boldsymbol{w}_{u}^{s_{t}T+z}||] = \frac{\eta_{i_{c}}^{2} (U+1)}{U} ([ (2 C_{1} + T(T+1)) \sum_{z=1}^{T} \sum_{u=1}^{U} || g_{u}^{s_{t}T+z}||^{2} ] + \frac{T(T+1)\sigma^{2}}{B})$,
where \textcircled{\footnotesize 1} follows from the fact that the terms $||g_{l}||^{2}$ are positive. Now, taking summation over global communication rounds gives: $\sum_{s_{t}=1}^{I/T-1} \sum_{r=1}^{T} \sum_{u=1}^{U} [\mathbb{E} ||\bar{\boldsymbol{w}}^{s_{t}T+z} - \boldsymbol{w}_{u}^{s_{t}T+z}||] = \frac{\eta_{i_{c}}^{2} (U+1)}{U} ( [ (2 C_{1} + T(T+1)) \sum_{i=1}^{I} \sum_{u=1}^{U} || g_{u}^{i}||^{2} ] + \frac{I(T+1)\sigma^{2}}{B} )$,
which leads to $\frac{1}{I} \sum_{i=1}^{I} \sum_{u=1}^{U} [\mathbb{E} ||\bar{\boldsymbol{w}}^{i} - \boldsymbol{w}_{u}^{i}||] \stackrel{\textcircled{\footnotesize 1}}{\leq} \frac{(2 C_{1} + T(T+1))}{I} \frac{\lambda \eta_{i_{c}}^{2} (U+1)}{U} \sum_{i=0}^{I-1} || \sum_{u=1}^{U} g_{u}^{i}||^{2} + \frac{\eta_{i_{c}}^{2} I(U+1)(T+1)\sigma^{2}}{UB}$, where \textcircled{\footnotesize 1} follows from the definition of weighted gradient diversity and the upper bound assumption. Now, $\frac{1}{KT} \sum_{k=1}^{K} \sum_{t=1}^{T} \sum_{u=1}^{U} [\mathbb{E} ||\bar{\boldsymbol{w}}_{k}^{t} - \boldsymbol{w}_{u,k}^{t}||] \leq \frac{(2 C_{1} + T(T+1))}{KT} \frac{\lambda \eta_{i_{c}}^{2} (U+1)}{U} \sum_{k=1}^{K} \sum_{t=1}^{T} || \sum_{u=1}^{U} g_{u,k}^{t}||^{2} + \frac{\eta_{i_{c}}^{2} KT(U+1)(T+1)\sigma^{2}}{UB}$. We note that,
\begin{align}
    &\frac{1}{KT} \sum_{k=1}^{K} \sum_{t=1}^{T} \sum_{u=1}^{U} \bigg[||\bar{\boldsymbol{w}}_{k}^{t} - \boldsymbol{w}_{u,k}^{t}||^{2}\bigg] \nonumber \\
    &= \frac{1}{KT} \sum_{k=1}^{K} \sum_{t=1}^{T} \sum_{u=1}^{U} \bigg[\underbrace{||\bar{\boldsymbol{w}}_{k}^{t} - \bar{\mathbf{v}}_{k}^{t}||^{2}}_{A_{1}}\bigg] \nonumber \\
    &+ \frac{1}{KT} \sum_{k=1}^{K} \sum_{t=1}^{T} \sum_{u=1}^{U} \bigg[\underbrace{||\bar{\mathbf{v}}_{k}^{t} - \boldsymbol{w}_{u,k}^{t} ||^{2}}_{A_{2}}\bigg] \nonumber \\
    &+ \frac{1}{KT} \sum_{k=1}^{K} \sum_{t=1}^{T} \sum_{u=1}^{U} \bigg[\underbrace{2 \langle \bar{\boldsymbol{w}}_{k}^{t} - \bar{\mathbf{v}}_{k}^{t}, \bar{\mathbf{v}}_{k}^{t} - \boldsymbol{w}_{u,k}^{t}  \rangle}_{A_{3}}\bigg].
\end{align}
If expectation is taken over $\mathcal{M}$, $(A_{3})$ vanishes due to unbiasedness of $\bar{\boldsymbol{w}}_{k}^{t}$. If it is not a global communication round, $A_{1}$ also vanishes as $\bar{\boldsymbol{w}}_{k}^{t} = \bar{\mathbf{v}}_{k}^{t}$. We use Lemma \ref{lemma3} to bound $A_{2}$. $\frac{1}{KT} \sum_{k=1}^{K} \sum_{t=1}^{T} \sum_{u=1}^{U} \bigg[\mathbb{E} ||\bar{\boldsymbol{w}}_{k}^{t} - \boldsymbol{w}_{u,k}^{t}||^{2}\bigg] \leq \frac{(2 C_{1} + T(T+1))}{KT} \frac{\lambda \eta_{i_{c}}^{2} (U+1)}{U} \sum_{k=0}^{K-1} \sum_{t=0}^{T-1} || \sum_{u=1}^{U} g_{u,k}^{t}||^{2} + \frac{\eta_{i_{c}}^{2} KT(U+1)(T+1)\sigma^{2}}{UB}.$

For global round, $\bar{\boldsymbol{w}}_{k}^{t}$ and $\bar{\mathbf{v}}_{k}^{t}$ are equivalent to $\bar{\boldsymbol{w}}_{k+1}$ and $\bar{\mathbf{v}}_{k+1}$. Hence, we additionally use Lemma \ref{lemma3} to bound $A_{1}$. Then $\frac{1}{KT} \sum_{k=1}^{K} \sum_{t=1}^{T} \sum_{u=1}^{U} [\mathbb{E} ||\bar{\boldsymbol{w}}_{k}^{t} - \boldsymbol{w}_{u,k}^{t}||^{2}] \leq \frac{4}{M} \eta_{k}^{2} T^{2} G^{2} + \frac{(2 C_{1} + T(T+1))}{KT} \frac{\lambda \eta_{i_{c}}^{2} (U+1)}{U} \sum_{k=1}^{K} \sum_{t=1}^{T} || \sum_{u=1}^{U} g_{u,k}^{t}||^{2} + \frac{\eta_{i_{c}}^{2} KT(U+1)(T+1)\sigma^{2}}{UB}$.

\subsection{Proof of Theorem 1} \label{prooftheorem1}
\noindent
Using Lemma \ref{lemma1} and Lemma \ref{lemma4}, we continue to further upper bound the inequality in (13) as follows:
\begin{align}
    &\frac{1}{KT} \sum_{k=1}^{K} \sum_{t=1}^{T} \mathbb{E}[f(\bar{\boldsymbol{w}}_{k}^{t+1}) - f(\bar{\boldsymbol{w}}_{k}^{t})] \nonumber \\
    &\leq \frac{1}{KT} \sum_{k=1}^{K} \sum_{t=1}^{T} \Bigg(-\eta_{k} \mathbb{E}\bigg[\langle \nabla f(\bar{\boldsymbol{w}}_{k}^{t}), \Tilde{g}_{k}^{t} \rangle\bigg]\Bigg) \nonumber \\
    & \hspace{6em} + \frac{1}{KT} \sum_{k=1}^{K} \sum_{t=1}^{T} \frac{\eta_{k}^{2}  L}{2} \mathbb{E}\bigg[||\Tilde{g}_{k}^{t}||^{2}\bigg] \nonumber \\
    &= \frac{1}{KT} \sum_{k=1}^{K} \sum_{t=1}^{T} \bigg(-\frac{\eta_{k}}{2}||\nabla f(\bar{\boldsymbol{w}}_{k}^{t})||^2 - \frac{\eta_{k}}{2}||\sum_{u=1}^{U}\nabla f_{u}(\boldsymbol{w}_{u,k}^{t})||^{2}\bigg) \nonumber \\
    &+ \frac{\lambda \eta_{k} L^{2}}{2KT} \frac{U+1}{U} \Bigg(\lambda \bigg[2C_{1}+T(T+1)\bigg] \eta_{k}^{2} \frac{1}{KT} \sum_{k=1}^{K} \sum_{t=1}^{T}||^2 \nonumber \\
    &- \frac{\eta_{k}}{2}||\sum_{u=1}^{U}\nabla f_{u}(\boldsymbol{w}_{u,k}^{t})||^{2} \Bigg) \nonumber \\
    &+ \frac{KT(L+1)\eta_{k}^{2} \sigma^{2}}{B} + \frac{1}{KT} \sum_{k=1}^{K} \sum_{t=1}^{T} \frac{\lambda L \eta_{k}^{2}}{2} \lambda \bigg(\frac{C_{1}}{U}+1\bigg) \nonumber \\
    &\bigg[||\sum_{u=1}^{U}\nabla f_{u}(\boldsymbol{w}_{u,k}^{t})||^{2}\bigg] + \frac{L \eta_{k}^{2}}{2} \frac{\sigma^{2}}{UB} + \frac{\text{var}(\xi_{u,k}^{t})}{M}. \label{39} 
\end{align}
Now, we have
\begin{align}
    &\frac{1}{KT} \sum_{k=1}^{K} \sum_{t=1}^{T} \mathbb{E}[f(\bar{\boldsymbol{w}}_{k}^{t+1}) - f(\bar{\boldsymbol{w}}_{k}^{t})] \nonumber \\
    &\stackrel{\textcircled{\footnotesize 1}}{\leq} - \frac{1}{KT} \sum_{k=1}^{K} \sum_{t=1}^{T} \frac{\eta_{k}}{2}||\nabla f(\bar{\boldsymbol{w}}_{k}^{t})||^2 + \frac{\eta_{k}^{3} L^{2} (T+1) \sigma^{2}}{B} \nonumber \\
    &\times \bigg(\frac{U+1}{U}\bigg) + \frac{L \eta_{k}^{2}}{2} \frac{\sigma^{2}}{UB} + \frac{\text{var}(\xi_{u,k}^{t})}{M}, \label{eqn40}
\end{align}
where \textcircled{\footnotesize 1} follows if the following condition holds: $-\frac{\eta_{k}}{2} + \frac{\lambda (U+1) L^{2} \eta_{k}^{3} [2C_{1}+T(T+1)]}{2U} + \frac{\lambda L \eta_{k}^{2}}{2} \bigg(\frac{C_{1}}{U}+1\bigg) \leq 0$.
Thus, it can be calculated to reveal that
\begin{align}
        &\frac{1}{KT} \sum_{k=1}^{K} \sum_{t=1}^{T} \mathbb{E}||\nabla f(\bar{\boldsymbol{w}}_{k}^{t})||^{2} \leq \frac{2 [f(\bar{\boldsymbol{w}}_{1}^{0}) - f^{*}]}{\eta_{k} KT} + \frac{L \eta \sigma^{2}}{UB} \nonumber \\
        &+ \frac{2 \eta_{k}^{2} \sigma^{2} L^{2} (T+1)}{B} \bigg(1+\frac{1}{U}\bigg) + \frac{1}{U} \sum_{u \in \mathcal{U}} \frac{\nu N_{z} D Tr(Z^{2})}{2H} \nonumber \\
        &\hspace{17em}+ C, \label{eqn79}
\end{align}
\textcolor{black}{where $C \propto \frac{1}{M} T^{2} \eta_{k_{0}}^{2} G^{2}$ is the additional term resulting from the partial device participation. Specifically, $C$ represents the aggregation mismatch between the ideal full-participation update (aggregating all $U$ devices) and the practical update obtained from only $M$ selected devices per round. The inverse relationship on $M$ indicates that choosing fewer devices increases this penalty due to stronger client sub-sampling effects. The quadratic dependence on $T$ and $\eta_{k_0}^{2}$ denotes the accumulation of local drift over $T$ local steps before synchronization, and larger step sizes amplify this effect. Finally, $G^{2}$ represents the sensitivity of the bound to the magnitude of local update directions. The larger gradient norms yield a larger worst-case penalty. Therefore, $C$ formalizes the communication-efficiency versus convergence-accuracy trade-off in DQL. Also, this term bounds the convergence rate of DQL under partial device participation.}

\textcolor{black}{The term $\frac{1}{U} \sum_{u \in \mathcal{U}} \frac{\nu N_{z} D Tr(Z^{2})}{2H}$ represents the direct impact of the quantum physical layer. It explicitly couples the number of measurement shots ($H$) and calibration noise ($\nu$) with the properties of the observable ($Tr(Z^2)$). This demonstrates that the global convergence trajectory is inherently restricted by the statistical uncertainty of wavefunction collapse.}

\vspace{-5em}
\begin{IEEEbiographynophoto}{Atit Pokharel} is a Ph.D. student in the Department of Electrical and Computer Engineering at The University of Alabama in Huntsville. His research focuses on secure quantum computing systems, quantum computing and quantum sensing with several top-tier IEEE publications.
\end{IEEEbiographynophoto}
\vspace{-5em}
\begin{IEEEbiographynophoto}{Shaba Shaon} is currently pursuing her Ph.D. in the Department of Electrical and Computer Engineering at The University of Alabama in Huntsville, Huntsville, AL, USA. She is with the Networking, Intelligence, and Security Research Lab at the same institution.  Her research interests focus on distributed quantum computing and quantum sensing.

\end{IEEEbiographynophoto}
\vspace{-5em}
\begin{IEEEbiographynophoto}{Thomas Morris} is a professor at the Department of Electrical and Computer Engineering, The University of Alabama in Huntsville (UAH), USA. He is the director of Center for Cybersecurity Research and Engineering (CCRE) at UAH. His current research involves security for industrial control systems.


\end{IEEEbiographynophoto}

\begin{IEEEbiographynophoto}{Dinh C. Nguyen} is an assistant professor at the Department of Electrical and Computer Engineering, The University of Alabama in Huntsville, USA.  He obtained the Ph.D. degree in computer science from Deakin University, Australia in 2021. His  research focuses on quantum computing, federated  learning,  wireless networking, and security.   

\end{IEEEbiographynophoto}

\end{document}